\def\arxiv{arxiv}
\def\elsevier{elsevier}

\def\version{arxiv} 

\ifx\version\arxiv 
\documentclass[final,5p,times,twocolumn]{elsarticle}
\fi

\ifx\version\elsevier 
\documentclass[final,5p,times,twocolumn]{elsarticle}
\fi

\usepackage{lscape}

\usepackage{amsmath,amsthm,amsfonts, amssymb}
\allowdisplaybreaks
\usepackage{algorithm} 
\usepackage{algorithmic}
\usepackage{enumerate}
\usepackage{bm}

\usepackage{nomencl}
\makenomenclature

\usepackage{url}
\usepackage{graphicx}
\graphicspath{{./fig/}}
\usepackage{subcaption}

\usepackage{framed} 
\usepackage{multicol} 
\usepackage{nomencl} 
\makenomenclature
\renewcommand*\nompreamble{\begin{multicols}{2}}
\renewcommand*\nompostamble{\end{multicols}}

\usepackage{algorithm,algorithmic}%
\usepackage{xcolor}

\theoremstyle{plain}
\newtheorem{thm}{Theorem}

\newtheorem{prop}{Proposition}

\newtheorem{assumption}{Assumption}

\theoremstyle{definition}
\newtheorem{defn}{Definition}

\theoremstyle{remark}
\newtheorem{rem}{Remark}

\usepackage{float}

\usepackage{tikz}
\usetikzlibrary{arrows.meta}
\tikzset{%
  >={Latex[width=2mm,length=2mm]},
            base/.style = {rectangle, rounded corners, draw=black,
                           minimum width=4cm, minimum height=1cm,
                           text centered, font=\sffamily},
        scenario/.style = {base, fill=blue!30},
          robust/.style = {base, fill=red!30},
         compute/.style = {base, fill=orange!15},
         other/.style = {base, fill=black!15},
}

\newcommand{\XG}[1]{{\textcolor{red}{{\textsf{ {\footnotesize (XB's comment: #1)}}}}} }

\newcommand{\FirstRound}[1]{{\textcolor{blue}{#1}}}

\begin{document}

\begin{frontmatter}



\title{Two-stage Robust Energy Storage Planning with Probabilistic Guarantees: A Data-driven Approach}

\author[XJTU]{Chao Yan\fnref{XJTU-Funding}\fnref{Contribution}}
\ead{yanchao19911224@stu.xjtu.edu.cn}
\author[TAMU]{Xinbo Geng\fnref{TAMU-Funding}\fnref{Contribution}}
\ead{xbgeng@tamu.edu}
\author[XJTU]{Zhaohong Bie\fnref{XJTU-Funding}}
\ead{zhbie@xjtu.edu.cn}
\author[TAMU]{Le Xie\fnref{TAMU-Funding}\corref{Cor}}
\ead{le.xie@tamu.edu.}

\fntext[Contribution]{C. Yan and X. Geng contributed equally to this paper.}
\fntext[XJTU-Funding]{C. Yan and Z. Bie's work is supported by China NSF U1965103.}
\fntext[TAMU-Funding]{X. Geng and L. Xie's work is supported by US NSF CCF-1934904.}
\cortext[Cor]{Corresponding author.}
\address[XJTU]{Department of Electrical  Engineering, Xi'an Jiaotong University, Xi'an, China, 710049.}
\address[TAMU]{Department of Electrical and Computer Engineering, Texas A\&M University, College Station, TX, United States, 77843.}



\begin{abstract}
In conventional planning decision making, shorter-term (e.g., hourly) variations are not explicitly accounted for. However, given the deepening penetration of variable resources, it is becoming imperative to consider such shorter-term variation in the longer-term planning exercise. This paper addresses a central challenge of jointly considering such shorter-term and longer-term uncertainties in power system planning with increasing penetration of renewable and storage resources.  
By leveraging the abundant operational observation data, we propose a  scenario-based robust planning framework that provides rigorous guarantees on the future operation risk of planning decisions considering a broad range of operational conditions, such as  renewable generation fluctuations and  load variations.
By connecting two-stage robust optimization with the scenario approach theory, we show that with a carefully chosen number of  scenarios, the operational risk level of the robust solution can be adaptive to the risk preference set by planners. The theoretical guarantees hold true for any distributions, and the proposed approach is scalable towards real-world power grids. Furthermore, the column-and-constraint generation algorithm is used to solve the two-stage robust planning problem and tighten theoretical guarantees. We substantiate this framework through a planning problem of energy storage in a power grid with deep renewable penetration.  
Case studies are performed on large-scale test systems (modified IEEE 118-bus system) to illustrate the theoretical bounds as well as the scalability of the proposed algorithm. 
\end{abstract}



\end{frontmatter}

\section{Introduction}
\FirstRound{Power system planning refers to a decision making process typically involves a time span of multiple years, and is uncertain by nature \cite{planning_book}. Conventionally, power system planning mainly incorporates \emph{longer-term} uncertainties such as load growth and fuel prices.
With the rapid growth of renewable resources, \emph{shorter-term} uncertainties such as renewable fluctuations and load variations become imperative in the longer-term planning exercise \cite{planning_review_2017}.
If not accounted for, the short-term uncertainties will likely render an overly expensive or unreliable planning outcome in the long-term.}

\FirstRound{To facilitate the integration of rapidly growing renewable resources, energy storage is being deployed at an accelerated pace in  power systems.
From 2014 to 2019, the installed capacity of energy storage increased by 35.7\% from 24.6 GW to 33.4 GW in the United States \cite{DOE-storage-2013,us_department_of_energy_doe_nodate}. 
According to China Energy Storage Alliance (CNESA), 32.3 GW of energy storage has been installed in China as of 2019.
Furthermore, Wood Mackenzie recently predicted that energy storage is poised for a decade-defining boom, with capacity set to grow by almost 33\% worldwide every year in the 2020s to reach around 741 GWh by 2030 \cite{WM-2020}.
Consequently, the optimal siting and sizing of energy storage, i.e., storage planning, becomes pivotal to build carbon-neutral and reliable  power systems.}

To consider shorter-term uncertainties in longer-term power system planning, stochastic optimization (SO) and robust optimization (RO) are the most commonly adopted approaches.
SO is built on probabilistic modeling of uncertainties. Since accurate probabilistic models are often unavailable or expensive to obtain, a massive number of samples are needed \cite{fully_renewable_energy}. This, however, will significantly increase computational burden \cite{trans_sequential_planning}.
RO, on the other hand, relies on a set-based (often deterministic) modeling of uncertainties. \emph{Robust} planning models uncertainties by a pre-defined uncertainty set, which describes a range. Therefore, reducing the computation burden of planning problem. Common choices of uncertainty sets include intervals \cite{interval_minimum_load} and polyhedral uncertainty sets \cite{robust_transmission_planning} . However, this choice of these uncertainties is often scrutinized for returning overly conservative solutions. How to construct the reasonable uncertainty set of these short-term uncertainties such as load variation and renewable fluctuation in the long-term planning decision is the core difficulty of current robust power system planning (RSP) problems.  

\FirstRound{Some recent studies have been proposed to tackle this challenge in RSP, most of them rely on historical data to construct the uncertainty set.
For instances, \cite{stochastic_robust_planning} extracted representative scenarios from historical operation data (HOD) to model short-term operation uncertainties; \cite{adaptive_robust_2020} constructed adaptive uncertainty sets to consider the risk from wind uncertainties in RSP. In most of the proposed approaches (e.g., \cite{adaptive_robust_2020}), short-term uncertainties are represented by a single operating point, thus temporal correlations among uncertainties are neglected.
Since the principal value of energy storage devices come from offering temporal flexibilities, it is pivotal to model temporal correlations in the context of storage planning, 
Reference \cite{robust_storage_in} constructed uncertainty sets as the convex hull of historical wind and load (time-varying) profiles. Similarly, \cite{operational_robust_planning} built polyhedral uncertainty sets for representative days by clustering HOD thus captures the temporal dynamics of uncertainties.}

\FirstRound{In this paper, we demonstrate that a simple construction of uncertainty sets, i.e., the collection of a carefully chosen number of i.i.d. scenarios, possesses many desirable features and advantages. By tuning a risk parameter (violation probability $\overline{\epsilon}$) and the associated number of scenarios ($K$), planning decisions can be adapted to risk preferences. Meanwhile, the scenarios can be daily operating conditions considering the temporal couplings. This simple construction of uncertainty set is closely related with the scenario approach theory \cite{campi_exact_2008,calafiore_random_2010,campi_wait-and-judge_2016,campi_general_2018}. The scenario approach has been successfully applied to many power system problems, e.g., resource adequacy and security assessment \cite{reserve_security_framework}, economic dispatch \cite{scenario_dispatch}, demand response scheduling \cite{scenario_demand_response}, and unit commitment \cite{scenario_unit,geng_computing_2019}.}

The scenario approach theory is a classical data-driven mathematical program theory with rigorous probabilistic guarantees \cite{campi_exact_2008,calafiore_random_2010,campi_wait-and-judge_2016,campi_general_2018} and also has been applied to power system field \cite{reserve_security_framework,scenario_dispatch,scenario_unit}. In this paper, it is developed with the two-stage robust program to address the above short-term uncertainty modeling issues in the power system planning. It is the first paper applying the scenario approach to two-stage robust power system planning. Specifically, we show that a simple construction of uncertainty sets, i.e., the collection of a carefully chosen number of i.i.d. operational scenarios, could resolve the issue of conservative solutions of the robust optimization . By tuning a risk parameter (violation probability $\epsilon$) and the associated number of scenarios ($K$), planning decisions can be adapted to risk preferences. Meanwhile, the scenarios can be daily operational conditions considering temporal operational conditions. Finally, we addressed two very important research problems that have never been successfully resolved in the existing research: (i) Why and how can we  leverage the historical operational data to represent the operational risk in the long-term power system planning problem based on a rigorous mathematical theory? (ii) How can we quantitatively control the operational risk considered in the robust power system planning decision through simply adjusting the data-based uncertainty set of short-term uncertainties, and make a trade-off between the robustness, investment costs and operational risks for the planning of storage system needing the consideration of temporal operational conditions?   After resolving these problems, we makes the following contributions:
 
(1) We propose a novel two-stage robust storage planning framework to facilitate the integration of renewables based on the scenario approach. Distribution-free theoretical guarantees are provided on the operational risks of the planning solutions of convex and non-convex robust storage planning models for future short-term uncertainties. These operational risks include  (a) future operation cost risk; and (b) the load loss risk.

(2) Although the scenario approach is commonly used for single-stage decision making, we establish the connection between two-stage RO and the scenario approach theory; the theoretical guarantees on operational risk are obtained via this connection by examining the cardinality of invariant sets.  We further proposes to use C\&CG algorithm to find invariant set numerically thus improve the theoretical risk guarantees.

(3) We leverage the theoretical risk guarantees and randomized property from data to make a balance between the robustness, investment costs and operational risk in the robust storage planning problem.

(4) The proposed storage planning approach is tested on the IEEE 118-bus system with realistic wind and load data obtained from Electric Reliability Council of Texas (ERCOT). Numerical results show that the cardinality of invariant sets of convex or non-convex two-stage robust storage planning (RSP)  problems is always small, regardless of problem size or power system scale. Consequently, we can achieve the same guarantee on future risk using a moderate number of scenarios, making the proposed approach computationally scalable.

\FirstRound{The proposed RSP framework in this paper lies at the intersection among four areas (Figure \ref{fig:overview}) and makes unique contributions in both methodology and domain application perspectives. The proposed framework is computationally efficient, exploits the value from pervasive operation data, and provides generic and rigorous guarantees on planning solutions.}

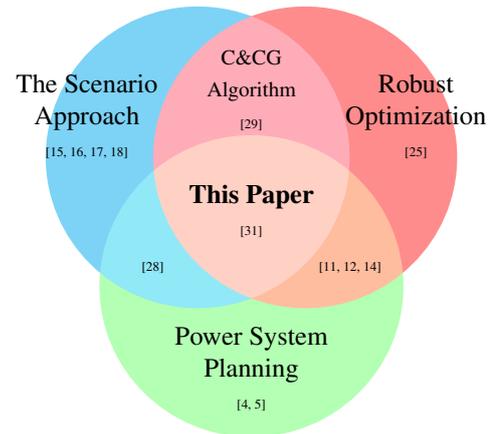
\begin{figure}[htbp]
\centering
\begin{tikzpicture}
  \begin{scope}[blend group = soft light]
    \fill[red!45!white]   (45:1.0) circle (2);
    \fill[cyan!45!white] (135:1.0) circle (2);
    \fill[green!30!white]  (270:1.0) circle (2);
  \end{scope}
  \node at (30:2.5) [align=center]  {Robust\\Optimization\\\tiny{\cite{ben-tal_robust_2009}}};
  \node at (150:2.5) [align=center] {The Scenario\\Approach\\\tiny{\cite{campi_exact_2008,calafiore_random_2010,campi_wait-and-judge_2016,campi_general_2018}}};
  \node at (270:2.1) [align=center]  {Power System\\Planning\\\tiny{\cite{planning_book,planning_review_2017}}};
  \node [align=center] {\textbf{This Paper}\\\tiny{\cite{yan_two-stage_2021}}};
  \node at (210:1.5) [font=\tiny, align=center]  {\cite{yan_scenario-based_2019}};
  \node at (330:1.5) [font=\tiny, align=center]  {\cite{robust_transmission_planning,stochastic_robust_planning,adaptive_robust_2020}};
  \node at (90:1.6) [font=, align=center]  { {\footnotesize C\&CG}\\{\footnotesize Algorithm} \\{\tiny \cite{zeng2013solving}} };
\end{tikzpicture}
\caption{This paper connects power system planning, the scenario approach, robust optimization, and C\&CG algorithm. A more detailed review is in \cite{geng_data-driven_2019-2}.}
\label{fig:overview}
\end{figure}

The notations in this paper are standard. All matrices and vectors are in the real field $\mathbb{R}$.
Matrices and vectors are in bold fonts, e.g., $\bm{A}$ and $\bm{b}$. The transpose of a vector $\bm{a}$ is $\bm{a}^\intercal$, and the $j$th entry of vector $\bm{a}$ is $a_j$. 
Sets are in calligraphy fonts, e.g. the set of transmission lines $\mathcal{L}$.
\FirstRound{A set consisting of $K$ elements $\bm{\delta}^{(k)},\cdots,\bm{\delta}^{(K)}$ is denoted by $\{\bm{\delta}^{(k)},\cdots,\bm{\delta}^{(K)}\}$ or $\{\bm{\delta}^{(k)}\}_{k=1}^{|\mathcal{K}|}$ in short.}
The cardinality of a set $\mathcal{S}$ is $|\mathcal{S}|$.
Related variables are represented by the same letter but distinct superscripts. For instance, $p_{l,t}^{\mathcal{L}}$ denotes the line flows, and $p_{i,n,t}^{\mathcal{G}}$ is the generation output.
The upper and lower bounds on variable $\bm{a}$ are denoted by $\overline{\bm{a}}$ and $\underline{\bm{a}}$, respectively.

\FirstRound{The remainder of this paper is organized as follows. 
Section \ref{sec:storage_planning} provides complete details on the deterministic energy storage planning problem.
Section \ref{sec:robust_storage_planning} introduces two-stage robust optimization and four robust storage planning formulations being studied in this paper. 
Main theoretical results are presented in Section \ref{sec:theoretical_analysis}. A-priori and a-posteriori probabilistic guarantees are provided for optimal solutions to RSP problems in Sections \ref{sub:prior_guarantees_RSP}-\ref{sub:posterior_guarantees_RSP}. Section \ref{sub:column_and_constraint_generation_algorithm} shows that the C\&CG algorithm is able to finding invariant sets while solving RSP problems.
Section \ref{sec:case_study} presents numerical results. Concluding remarks and future work are in Section \ref{sec:concluding_remarks}.
\ifx\version\arxiv
All proofs and detailed algorithms are provided in the appendices.
\else
All proofs and detailed algorithms are provided in the full-length version of this paper on arXiv \cite{yan_two-stage_2021}.
\fi
}

\section{Storage Planning} 
\label{sec:storage_planning}
Energy storage is a key resource to facilitate the integration of renewable energy resources by providing operational flexibility and ancillary services \cite{xu_scalable_2017}. As of 2019, PJM has deployed approximately 300 MW of energy storage \cite{hsia_energy_2019}; about 20 MW grid-scale battery-storage projects have been online in ISO New England since 2015, and nearly 2300 MW of grid-scale stand-alone energy-storage projects are in the queue to be interconnected \cite{iso_new_england_2020_2020}.
Appropriate placements of energy storage systems could maximize the value of storage on the secure and economic operation of power systems. 
Storage planning, decising the optimal siting and sizing of energy storage\cite{storage_size_site,storage_technology}, becomes a topic of increasing importance in power system planning. Similar with other planning problems, various approaches have been proposed to deal with uncertainties, e.g., stochastic storage planning \cite{stochastic_storage_planning}, robust optimization planning \cite{robust_storage_in}, and chance-constrained planning \cite{chance_storage_planning}. As storage is a special type of resources providing temporal flexibility, the consideration of temporal operational dynamics of short-term uncertainties is crucial in the planning exercise to effectively reduce future risk.

\subsection{Nomenclature} 
\label{sub:nomenclature}
For notation simplicity, we use $(\cdots)\text{s}$ to denote the collection of variables, e.g., $\{(E_{n}, P_{n}, z_n)\}_{n \in \mathcal{N}}$ is denoted by $(E_{n}, P_{n}, z_n)\text{s}$. Unless specified, all decision variables are continuous. For simplicity, we assume that storage systems can be installed at every bus, i.e., $\mathcal{S} = \mathcal{N}$, thus all variables related with storage systems is indexed by the bus index $n$, e.g., $e_{n,t}^{\text{SOC}}$.

\ifx\version\arxiv
\begin{table*}[!t]
   \begin{framed}
     \printnomenclature
   \end{framed}
\end{table*}
\nomenclature{$i,j,k$}{indices for general purposes;}
\nomenclature{$s,t,l,n$}{ indices for storages, time, lines and nodes}
\nomenclature{$\mathcal{G}(n)$}{ the set of generators at bus $n$;}
\nomenclature{$\mathcal{T}$}{ the set of time;}
\nomenclature{$\mathcal{N}$}{ the set of buses;}
\nomenclature{$\mathcal{L}$}{ the set of transmission lines;}
\nomenclature{$\mathcal{S}$}{ potential locations of storage systems;}
\nomenclature{$o(l), r(l)$}{ sending and receiving nodes of a line, i.e., $l = (o(l),r(l))$;}
\nomenclature{$\mathcal{W}$}{ the set of wind farms.}
\nomenclature{$E_{n}$}{ the energy capacity (MWh) of the storage system;}
\nomenclature{$P_{n}$}{ the power capacity (MW) of the storage system;}
\nomenclature{$z_{n}$}{ integer variables, number of storage units.}
\nomenclature{$p_{n,t}^{\text{ch}}, p_{n,t}^{\text{dis}}$}{ the charging and discharging power (MW) of the storage system at node $n$;}
\nomenclature{$p_{i,n,t}^{\mathcal{G}}$}{ generator output;}
\nomenclature{$p_{l,t}^{\mathcal{L}}$}{ line power flow;}
\nomenclature{$\theta_{n,t}$}{ nodal voltage angle;}
\nomenclature{$e_{n,t}^{\text{SOC}}$}{ state of charge (SOC) of a storage;}
\nomenclature{$p_{n,t}^{\text{shed}}$}{ nodal load curtailment/shedding.}
\nomenclature{$C_{n}^{\text{E}}$}{ annualized storage energy investment cost;}
\nomenclature{$C_{n}^{\text{P}}$}{ annualized storage power investment cost;}
\nomenclature{$c_i^{\mathcal{G}}$}{ hourly incremental generation cost;}
\nomenclature{$c_{n}^{\text{dis}}, c_{n}^{\text{ch}}$}{ hourly incremental discharging/charging cost, related with the degradation cost of storage units;}
\nomenclature{$c_n^{\text{shed}}$}{ cost of load curtailment;}
\nomenclature{$\underline{p_i^{\mathcal{G}}}, \overline{p_i^{\mathcal{G}}}$}{ generator output lower/upper limit;}
\nomenclature{$\underline{p_l^{\mathcal{L}}}, \overline{p_l^{\mathcal{L}}}$}{ transmission power flow lower/upper limit;}
\nomenclature{$\eta_{n}^{\text{ch}}, \eta_{n}^{\text{dis}}$}{ storage charge/discharge efficiency;}
\nomenclature{$C^{\text{Budget}}$}{ storage investment budget;}
\nomenclature{$x_l$}{ transmission line reactance;}
\nomenclature{$\overline{\rho}, \underline{\rho}$}{ maximum/minimum power/energy ratio;}
\nomenclature{$\overline{R}_{i}, \underline{R}_{i}$}{ generator ramp up/down limit;}
\nomenclature{$K_d$}{ the weight of operational cost;}
\nomenclature{$v_{n,t}$}{ binary variable, indicating charging/discharging status of storage devices;}
\nomenclature{$q^{E}$}{ quantized energy capacity;}
\nomenclature{$q^{P}$}{ quantized power capacity;}
\nomenclature{$\overline{w_{n}}$}{ maximum wind capacity at node $n$;}
\nomenclature{$\overline{d_{n}}$}{  peak load at node $n$.}
\nomenclature{$\bm{\delta}_{s}$}{ uncertainties in storage planning, $\bm{\delta_{s}} := (\alpha^w_{n,t},\alpha^d_{n,t})\text{s}$;}
\nomenclature{$\alpha^w_{n,t}$}{ wind capacity factor at bus $n$ at time $t$;}
\nomenclature{$\alpha^d_{n,t}$}{ load factor at bus $n$ at time $t$.}

\else
A complete list of variables is in the full-length version of this paper \cite{yan_two-stage_2021}.
\fi

\subsection{Deterministic Storage Planning} 
\label{sub:deterministic_storage_planning}
We first introduce the deterministic storage planning problem \eqref{opt:det-SP}, which is a two-stage optimization problem.
\begin{align}
\min_{ \bm{x} \in \Xi_{1}}~ C(\bm{x}) + K_d \min_{ \bm{y} \in \Xi_{2}(\bm{x}; \bm{\delta})} c(\bm{y})  \label{opt:det-SP} 
\end{align}
\subsubsection{Objective} 
\label{ssub:objective_DSP}
There are two common storage planning formulations: cost-minimizing formulation and curtailment-minimizing formulation. They only differ in objectives.
\paragraph{Cost-minimizing Storage Planning}
The objective of the cost-minimizing formulation is \eqref{opt-det:cost}. The \emph{first stage} determines the optimal energy and power ratings $\bm{x} = (E_{n}, P_{n}, z_{n})$ of the storage units at bus $n$ by minimizing total investment cost \eqref{opt-det:planning-obj-cost}. 
The \emph{second stage} optimizes the daily operation cost \eqref{opt-det:operation-obj-2} of the planned storages $(p_{n,t}^{\text{ch}},p_{n,t}^{\text{dis}})$, generation $p_{i,n,t}^{\mathcal{G}}$, and load curtailment $p_{n,t}^\text{shed}$, i.e., $\bm{y} = (p_{n,t}^{\text{ch}},p_{n,t}^{\text{dis}},p_{i,n,t}^{\mathcal{G}},p_{n,t}^\text{shed},v_{n,t})$. 

\begin{subequations}
\label{opt-det:cost}
\begin{align}
& C(\bm{x}) := \sum_{n\in \mathcal{N}} (C_{n}^{\text{E}} E_{n} +  C_{n}^{\text{P}} P_{n}) \label{opt-det:planning-obj-cost}, \\
& c(\bm{y}) := \sum_{t\in \mathcal{T}} \sum_{n \in \mathcal{N}}  \sum_{i\in \mathcal{G}(n)} c_{i,n}^{\mathcal{G}} p_{i,n,t}^{\mathcal{G}} + \sum_{t\in \mathcal{T}} \sum_{n\in \mathcal{N}} (c_{n}^\text{ch} p_{n,t}^{\text{ch}} \nonumber \\
&\hspace{3cm}  +c_{n}^\text{dis} p_{n,t}^{\text{dis}}) +\sum_{t\in \mathcal{T}}  \sum_{n\in \mathcal{N}} c_{n}^\text{shed} p_{n,t}^\text{shed}. \label{opt-det:operation-obj-2}
\end{align}
\end{subequations}

The first term $C(\bm{x})$ in the objective represents the annualized investment of storage systems. Cost coefficients for the energy and power rating investments are denoted by $C_{n}^{\text{E}}$ (unit:\$ / MWh) and $C_{n}^{\text{P}}$ (unit:\$ / MW), which are converted to present values\footnote{Let $\sigma$ be the annual interest and $\Gamma$ be the investment period, e.g. 20 years. $C_{n}^{\text{E}}$ and $C_{n}^{\text{P}}$ are computed as \eqref{eqn:cost-calculation}, where $C_n^{E,\text{total}}$ and $C_n^{P,\text{total}}$ represent the total storage energy and power annual investment cost over $\Gamma$ years.
\begin{eqnarray}
C_{n}^{\text{E}} = C_{n}^{\text{E},\text{total}} \frac{\sigma (1+\sigma)^\Gamma}{(1+\sigma)^\Gamma-1},~C_{n}^{\text{P}} = C_{n}^{\text{P},\text{total}} \frac{\sigma (1+\sigma)^\Gamma}{(1+\sigma)^\Gamma-1}.  \label{eqn:cost-calculation}
\end{eqnarray}
}.
For conceptual simplicity, the second stage uses a typical day to represent the operation of a whole year, i.e. $K_d = 365$. It is straightforward to extend towards employing $D$ typical days to represent the annual operation cost $\sum_{i=1}^{D} K_i c_i(\bm{y})$.

\FirstRound{This paper studies large-scale energy storage investment at the transmission level, and assumes that storage investment cost (including the land and construction cost), scales linearly with storage power and energy ratings.
The marginal production cost coefficients of energy storage (e.g., from battery efficiency loss or degradation) $c_{n}^{\text{ch}}$ and $c_{n}^{\text{dis}}$ are assumed to be constants.
Note that different storage technologies usually have different values for those cost parameters.
Compressed air energy storage (CAES) has high investment costs for power ratings but low investment cost for capacity, the marginal production cost is negligible.
For lithium batteries (LiBES) batteries, their degradation largely depends on daily operations, the aging of LiBES can be modeled as marginal production cost.
Without loss of generality, this paper assumes that battery degradation leads to constant marginal costs in each charging and discharging cycle.}

\paragraph{Curtailment-minimizing Storage Planning}
Equation \eqref{opt-det:load-curtailment} is the objective function of the curtailment-minimizing storage planning formulation, which minimizes total load curtailment in the worst-case scenario by investing on storages in the first stage. More specifically, there will be no first-stage cost, i.e., $C(\bm{x})=0$ in  \eqref{opt:det-SP}. The storage planning formulation using \eqref{opt-det:load-curtailment} is termed curtailment-minimizing storage planning.
\begin{align}
c(\bm{y}) := \sum_{t\in \mathcal{T}}  \sum_{n\in \mathcal{N}}  p_{n,t}^\text{shed}.
\label{opt-det:load-curtailment}
\end{align}

\subsubsection{First-stage Constraints} 
Constraints for the \emph{first stage} $\bm{x} \in \Xi_1$ is defined below:
\begin{equation*}
\Xi_1 := \Big \{ (E_{n}, P_{n}, z_n)\text{s}~\text{that satisfy constraints}~\eqref{discrete-investment:variable}\eqref{opt-det:constr-first-stage}\eqref{opt-det:constr-budget} \Big \}
\end{equation*}
We consider the case where the capacities of storage units are quantized (instead of being continuous), parameters $q^{\text{E}}$ and $q^{\text{P}}$ denote the smallest quantized energy and power ratings for one energy storage unit.
\begin{equation}
 z_{n} \in \{0,1,2,\cdots,\overline{z}_{n}\},~n\in \mathcal{N}.  \label{discrete-investment:variable}
\end{equation}
Integer variable $z_n$ is the number of storage units to be installed at bus $n$, which is constrained by the total quantity limit of storage investment \eqref{q_limit}. 
As computed in \eqref{energy-discrete} and \eqref{power-discrete}, $E_n$ and $P_n$ represent the energy and power capacities of the storage system at node $n$. For some energy storage,  the size of power and energy rating of storage unit is fixed due to physical limitations.  Their investments need to be modeled by \eqref{opt-det:constr-first-stage}. For other energy storage  without physical limitations on size, their investments can be directly represented by $E_{n}$, $P_{n}$.   
\begin{subequations}
\label{opt-det:constr-first-stage}
\begin{align}
& \sum_{n \in \mathcal{N}} z_{n} \le \overline{z} \label{q_limit},\\
& q^{\text{E}} z_{n} = E_{n},~n\in\mathcal{N} \label{energy-discrete},\\
& q^{\text{P}} z_{n} = P_{n},~n\in\mathcal{N} \label{power-discrete}. 
\end{align}
\end{subequations}
Constraint \eqref{opt-det:constr-budget} limits total investment within budget.
\begin{equation}
\sum_{n\in \mathcal{N}} (C_{n}^{\text{E}} E_{n} + C_{n}^{\text{P}} P_{n}) \le C^{\text{budget}}. \label{opt-det:constr-budget} 
\end{equation}

\subsubsection{Second-stage Constraints} 
\label{ssub:second_stage_constraints}
Constraints for the \emph{second stage} $\bm{y} \in \Xi_2(\bm{x}; \bm{\delta})$ are the set of daily operational constraints. Note the feasible region $\Xi_2$ depends on the first-stage decision $\bm{x} = (E_{n}, P_{n})$ and uncertainties $\bm{\delta} = (\alpha^w_{n,t},\alpha^d_{n,t})$.
\begin{multline}
\Xi_2(E_{n}, P_{n}; \bm{\delta}) := \Big \{(p_{n,t}^{\text{ch}},p_{n,t}^{\text{dis}},p_{i,n,t}^{\mathcal{G}},p_{n,t}^\text{shed}, v_{n,t})\text{s}~\\
\text{that satisfy constraints}~\eqref{opt-det:operation-constraints}  \eqref{opt-det:storage-constraints}\eqref{eqn:not_the_same_time_storage} \Big \}.
\end{multline}

\FirstRound{The inter-temporal operations of energy storage are \eqref{opt-det:storage-constraints} and \eqref{eqn:not_the_same_time_storage}, e.g., the state of charge (SOC) update from $t-1$ to $t$ in \eqref{opt-det:constr-soc-evolving}. Since \eqref{opt-det:storage-constraints} focuses on short-term operations, the degradation of storage systems in each charging/discharging cycle is neglected.
The charging and discharging power are limited by \eqref{opt-det:constr-charging}-\eqref{opt-det:constr-discharging}.  We assume that the power limits $P_{n}$ for charging and discharging are the same.} SOC is limited to be within energy capacity \eqref{opt-det:constr-soc}.
\begin{subequations}
\label{opt-det:storage-constraints}
\begin{align}
& e_{n,t}^{\text{SOC}} - e_{n,t-1}^{\text{SOC}} = p_{n,t}^{\text{ch}}\eta_{n}^{\text{ch}} - p_{n,t}^{\text{dis}}/\eta_{n}^{\text{dis}}, \label{opt-det:constr-soc-evolving} \\
& 0 \le p_{n,t}^{\text{ch}} \le P_{n} \cdot v_{n,t},\label{opt-det:constr-charging} \\
& 0 \le p_{n,t}^{\text{dis}} \le P_{n}\cdot (1-v_{n,t}), \label{opt-det:constr-discharging}  \\
& 0 \le e_{n,t}^{\text{SOC}} \le E_{n},\label{opt-det:constr-soc} \\
& \hspace{3cm} ~n \in \mathcal{N},~t \in \mathcal{T}. \nonumber
\end{align}
\end{subequations}

\FirstRound{Binary variables $v_{n,t}$ are introduced to avoid charging and discharging at the same time. $v_{n,t}=1$ indicates the storage system at node $n$ is charging at $t$; $v_{n,t}=0$ when discharging: }
\begin{align}
\label{eqn:not_the_same_time_storage}
v_{n,t} \in \{0,1\},~n \in \mathcal{N},~t \in \mathcal{T}.
\end{align}

\FirstRound{For the storage system like CAES, it can charge and discharge at the same time, the binary variable $v_{n,t}$ can be relaxed to be continuous.}  

\FirstRound{This paper studies the storage investment in transmission system.
Constraint \eqref{opt-det:operation-constraints} models the secure operation of the transmission power system.
\eqref{opt-det:constr-balance} is the power balance at every node;
\eqref{opt-det:constr-generation}-\eqref{opt-det:constr-ramp} are the capacity and ramping limits of generators.
Because this paper studies storage planning at the transmission level, DC power flow equations \eqref{opt-det:DC-power-flow} are used. All transmission lines, transformers and phase shifters are modeled with a common branch model, consisting of a standard (AC) $\Pi$ transmission line model. DC power flow equations are obtained via linearizing the AC transmission line model.
Line flow limits are in \eqref{opt-det:constr-line}. Constraints on load shedding are \eqref{opt-det:constr-cut}.
The actual wind generation in \eqref{opt-det:operation-constraints} is the product of wind capacity $\overline{w_{n}}$ and the wind capacity factor $\alpha^w_{n,t}$. The actual load is the product of the peak load $\overline{d_{n}}$ and load factor $\alpha^d_{n,t}$.} 
\begin{subequations}
\label{opt-det:operation-constraints}
\begin{align}
& \sum_{i \in \mathcal{G}(n)} p_{i,n,t}^{\mathcal{G}} + p_{n,t}^{\text{dis}} - p_{n,t}^{\text{ch}} + \alpha^w_{n,t}\overline{w_{n}} + \sum_{l|n\in r(l)} p_{l,t}^{f} - \nonumber \\
& \hspace{3cm} \sum_{l|n \in o(l)} p_{l,t}^{f}  + p_{n,t}^{\text{shed}} =\alpha^d_{n,t} \overline{d_{n}},\label{opt-det:constr-balance}  \\
& \underline{p_i^{\mathcal{G}}} \le p_{i,n,t}^{\mathcal{G}} \le \overline{p_i^{\mathcal{G}}}, \label{opt-det:constr-generation}\\
& -\underline{R}_{i} \le p_{i,n,t}^{\mathcal{G}} - p_{i,n,t-1}^{\mathcal{G}} \le \overline{R}_{i}, \label{opt-det:constr-ramp}\\
& \underline{p_l^{\mathcal{L}}} \le p_{l,t}^{\mathcal{L}} \le \overline{p_l^{\mathcal{L}}}, \label{opt-det:constr-line} \\
& 0 \le p_{n,t}^{\text{shed}} \le \alpha^d_{n,t}\overline{d_{n}}, \label{opt-det:constr-cut}  \\
& p_{l,t}^{\mathcal{L}} = \frac{1}{x_l} (\theta_{o(l),t}-\theta_{r(l),t}), \label{opt-det:DC-power-flow} \\
& \hspace{2cm} i \in \mathcal{G}(n),~n \in \mathcal{N},~ l \in \mathcal{L},~ t \in \mathcal{T}. \nonumber
\end{align}
\end{subequations}


\section{Background} 
\label{sec:background}

\subsection{Two-stage Robust Optimization} 
\label{sub:two_stage_robust_optimization}
A standard two-stage robust optimization problem is \eqref{opt:two-stage-robust}.
\begin{equation}
\label{opt:two-stage-robust}
\min _{\bm{x} \in \mathcal{X} } \left(\bm{c}^{\intercal} \bm{x}+ \max _{\bm{\delta} \in \bm{\Delta}} \min _{\bm{y} \in \mathcal{Y}(\bm{x}; \bm{\delta})} \bm{d}^{\intercal} \bm{y}\right)
\end{equation}
It seeks the objective-minimizing solution $(\bm{x}^*, \bm{y}^*)$ for the worst scenario $\bm{\delta} \in \Delta$ in a pre-defined \emph{uncertainty set} $\Delta$. Without loss of generality
\ifx\version\arxiv
(see Appendix \ref{ssub:epigraph_formulation})
\else
\cite{yan_two-stage_2021}
\fi
, we assume deterministic objectives, i.e., no randomness associated with $\bm{c}$ and $\bm{d}$. Sets $\mathcal{X} \subseteq \mathbb{R}^{n_x}$ and $\mathcal{Y} \subseteq \mathbb{R}^{n_y}$ denote the constraints for the first and second stages, respectively.
Note that the second stage constraint $\mathcal{Y}(\bm{x}; \bm{\delta})$ is determined by the first stage decision $\bm{x}$ and uncertainty $\bm{\delta}$. We follow the convention that $\bm{d}^{\intercal} \bm{y} = +\infty$ if the second-stage problem $ \min_{\bm{y} \in \mathcal{Y}(\bm{x}; \bm{\delta})} \bm{d}^{\intercal} \bm{y}$ is infeasible. 
\ifx\version\arxiv 
\begin{defn}[Feasible Solution]
A tuple $(\bm{x},\bm{y})$ is a \emph{feasible solution} to \eqref{opt:two-stage-robust} if (i) $\bm{x} \in \mathcal{X}$; and (ii) $\bm{y} \in \mathcal{Y}(\bm{x}; \bm{\delta})$ for all $\bm{\delta} \in \Delta$. Equivalently, $(\bm{x},\bm{y})$ is feasible if it has a finite objective value.
\end{defn}
\begin{rem}[Optimal and Infeasible Solutions]
\label{rem:infeasible_solution}
We say that $(\bm{x}^*, \bm{y}^*)$ is an (globally) \emph{optimal solution} to \eqref{opt:two-stage-robust}, if $\bm{c}^\intercal \bm{x}^* + \bm{d}^\intercal \bm{y}^* \le \bm{c}^\intercal \bm{x} + \bm{d}^\intercal \bm{y}$ for all feasible solutions $(\bm{x}, \bm{y})$.
For any potentially better solution $(\bm{x}^\diamond, \bm{y}^\diamond)$ with $\bm{c}^\intercal \bm{x}^\diamond + \bm{d}^\intercal \bm{y}^\diamond < \bm{c}^\intercal \bm{x}^* + \bm{d}^\intercal \bm{y}^*$, there always exists $\bm{\delta}^\diamond \in \Delta$ such that $(\bm{x}^\diamond, \bm{y}^\diamond)$ is infeasible, i.e., $\bm{y}^\diamond \notin \mathcal{Y}(\bm{x}^\diamond; \bm{\delta}^\diamond)$.
\end{rem}
\fi

Throughout this paper, we only consider the case in which all constraints are linear. Matrix $\bm{Q}(\bm{\delta})$ is the \emph{recourse} matrix.
\begin{subequations}
\begin{align}
\mathcal{X} &:=\{\bm{x} \in \mathbb{R}^{n_x}: {\bm{A}\bm{x} \le \bm{b}} \}\\
\mathcal{Y}(\bm{x}; \bm{\delta}) &:= \left\{\bm{y} \in \mathbb{R}^{n_y}: \bm{G} \bm{y} \le \bm{h},  \bm{T}(\bm{\delta}) \bm{x} + \bm{Q}(\bm{\delta}) \bm{y} \le  \bm{r}(\delta) \right\}
\end{align}
\end{subequations}
In addition, we focus on the cases where the second-stage problem is feasible, which is formally defined as \emph{relatively complete recourse}. This is mainly for the purpose of simplifying theorems and algorithms. Most results in Sections \ref{sub:probabilistic_feasibility_guarantees_via_the_scenario_approach} and \ref{sec:theoretical_analysis} can be easily generalized towards situations \emph{without} relatively complete recourse. Detailed discussions in different contexts are provided accordingly.
\begin{defn}[Relatively Complete Recourse \cite{zeng2013solving}]
\label{def:relative}
Two-stage RO problem \eqref{opt:two-stage-robust} is said to have \emph{relatively complete recourse}, if for any $\bm{x} \in \mathcal{X}$ and $\bm{\delta} \in \Delta$, the second stage problem $\min _{\bm{y} \in \mathcal{Y}(\bm{x}; \bm{\delta})} \bm{d}^{\intercal} \bm{y}$ is feasible, i.e., $\mathcal{Y}(\bm{x}; \bm{\delta}) \ne \emptyset$.
\end{defn}
One critical observation is that \eqref{opt:two-stage-robust} can be formulated as a \emph{single-stage} robust problem\footnote{This observation may not be true for adjustable robust optimization problems, see Chapter 14 of \cite{ben-tal_robust_2009} and Section 6 of \cite{gorissen_practical_2015} for in-depth discussions.}. The following proposition lays the foundation of the main theoretical results (Theorems \ref{thm:apriori_guarantee_convex}, \ref{thm:aposteriori_guarantee_convex} and \ref{thm:aposteriori_guarantee_nonconvex}) of this paper.
\begin{prop}
\label{prop:two-stage-RO-as-one-stage}
The two-stage robust optimization problem \eqref{opt:two-stage-robust} is \emph{equivalent} to the single-stage problem \eqref{opt:two-stage-robust-single}:
\begin{subequations}
\label{opt:two-stage-robust-single}
\begin{align}
\min_{\bm{x} \in \mathcal{X},\gamma}~&\bm{c}^{\intercal} \bm{x} + \gamma \\
\text{s.t.}~& (\bm{x},\gamma) \in \mathcal{Z}(\bm{\delta}) ,~\forall \bm{\delta} \in \Delta.
\end{align}
\end{subequations}
in which $\mathcal{Z}(\bm{\delta}) := \{(\bm{x},\gamma): \exists \bm{y} \in \mathcal{Y}(\bm{x}; \bm{\delta})~\text{and}~\bm{d}^{\intercal} \bm{y} \le \gamma\}$.
We say \eqref{opt:two-stage-robust} is equivalent with \eqref{opt:two-stage-robust-single} in the sense that they share the same optimal first-stage solutions and optimal objective values.
\end{prop}

The seemingly simple formulation \eqref{opt:two-stage-robust-single} may not be solved directly, the main reason is that set $\mathcal{Z}(\bm{\delta})$ could be complicated (e.g., the intersection of an exponential number of half-spaces), sometimes an analytical form of $\mathcal{Z}(\bm{\delta})$ may not even exist.
The only known property of $\mathcal{Z}(\bm{\delta})$ is its convexity if the original problem \eqref{opt:two-stage-robust} is convex\footnote{Essentially $\mathcal{Z}(\bm{\delta})$ is obtained by (1) lifting the original feasible region in $\mathbb{R}^{n_x} \times \mathbb{R}^{n_y}$ by introducing $\gamma$ then (2) projecting the feasible region onto $\mathbb{R}^{n_x}$. If the original feasible region is convex, then the affine projection of a convex set remains convex.}.
Proposition \ref{prop:two-stage-RO-as-one-stage} only aims at connecting two-stage RO problems with the single-stage scenario approach theory in Section \ref{sub:probabilistic_feasibility_guarantees_via_the_scenario_approach}.

\FirstRound{In the remainder of this paper, we make the choice that the uncertainty set $\Delta$ is a collection of $K$ i.i.d. realizations of random variables $\bm{\delta}$, i.e., $\Delta = \mathcal{K} := \{\bm{\delta}^{(1)},\bm{\delta}^{(2)},\cdots,\bm{\delta}^{(K)}\}$.} In Section \ref{sub:probabilistic_feasibility_guarantees_via_the_scenario_approach}, we show that this simple construction of uncertainty sets possesses rigorous theoretical guarantees. \ifx\version\arxiv  In addition, uncertainty set $\Delta = \mathcal{K}$ does not introduce any additional computational complexity (unlike ellipsoidal or conic uncertainty sets in \cite{bertsimas_tractable_2006}, which sometimes render solving RO problems intractable), and the resulting two-stage RO problem can be efficiently solved by C\&CG algorithm. \fi
\subsection{Column-and-Constraint Generation (C\&CG) Algorithm} 
\label{sub:column_and_constraint_generation_algorithm}
\FirstRound{One popular choice to solve the two-stage RO problem is the column-and-constraint generation (C\&CG) algorithm \cite{zeng2013solving}. When the uncertainty set is a collection of $K$ discrete scenarios $\mathcal{K} := \{\bm{\delta}^{(1)},\bm{\delta}^{(2)},\cdots,\bm{\delta}^{(K)}\}$,
\begin{equation}
\label{opt:two-stage-robust-scenarios}
\min _{\bm{x} \in \mathcal{X} } \left(\bm{c}^{\intercal} \bm{x}+ \max _{\bm{\delta} \in \mathcal{K} } \min _{\bm{y} \in \mathcal{Y}(\bm{x}; \bm{\delta})} \bm{d}^{\intercal} \bm{y}\right)
\end{equation}
Proposition \ref{prop:two-stage-RO-as-one-stage} shows that \eqref{opt:two-stage-robust-scenarios} is equivalent with 
\begin{subequations}
\label{opt:two-stage-robust-single-scenario}
\begin{align}
\min_{\bm{x},\gamma}~&\bm{c}^{\intercal} \bm{x} + \gamma \\
\text{s.t.}~& \bm{x} \in \mathcal{X}~\text{and}~(\bm{x},\gamma) \in \cap_{i=1}^{K} \mathcal{Z}(\bm{\delta}^{(i)}).
\end{align}
\end{subequations}
Constraint $(\bm{x},\gamma) \in \cap_{i=1}^{K} \mathcal{Z}(\bm{\delta}^{(i)})$ can be explicitly written as \eqref{opt:two-stage-robust-two-scenario} by introducing additional variables $\{\bm{y}^{(1)},\cdots, {\bm{y}^{(K)}}\}$. Variable $\bm{y}^{(k)}$ is the recourse decision variable for the $k$th scenario $\bm{\delta}^{(k)}$.
\begin{subequations}
\label{opt:two-stage-robust-two-scenario}
\begin{align}
&\bm{d}^{\intercal} \bm{y}^{(k)} \le \gamma,~k=1,2,\cdots,K. \\
& \bm{y}^{(k)}\in \mathcal{Y}(\bm{x}, \bm{\delta}^{(k)}),~k=1,2,\cdots,K.
\end{align}
\end{subequations}
}

\FirstRound{Therefore, the two-stage RO problem \eqref{opt:two-stage-robust} is equivalent to a single-stage optimization problem. It is worth pointing out that constraint \eqref{opt:two-stage-robust-two-scenario} is simply enumerating \emph{all} $K$ scenarios. When seeking risk-averse solutions, which is common for power system applications, the number of scenarios $K$ could be colossal. Thus \eqref{opt:two-stage-robust-two-scenario} might consist of a gigantic number of decision variables and constraints, which is extremely inefficient or even impossible to solve. A \emph{partial} enumeration such as the C\&CG algorithm could significantly outperform the approach of solving \eqref{opt:two-stage-robust-two-scenario} in one shot.}

\FirstRound{The intuition behind the C\&CG algorithm is quite simple: only a small portion $\mathcal{O}$ of the uncertainty set $\Delta$ matters, e.g., extreme points of $\Delta$ along the optimization direction. C\&CG algorithm is essentially an iterative procedure to identify critical scenarios $\mathcal{O}$. The C\&CG algorithm iteratively adds constraints $(\bm{x},\gamma) \in \cap_{i \in \mathcal{O}} \mathcal{Z}(\bm{\delta}^{(l)})$ to the problem (\emph{constraint generation}). Since constraint $(\bm{x},\gamma) \in \mathcal{Z}(\bm{\delta}^{(l)})$ guarantees the existence of feasible recourse variable $\bm{y}^{(l)} \in \mathcal{Y}(\bm{x}; \bm{\delta})$, additional second-stage variables are introduced (\emph{column generation}).
\ifx\version\arxiv 
Algorithm \ref{alg:CCG} in \ref{sub:algorithms} formally defines the C\&CG algorithm.
More details and theoretical analysis on C\&CG can be found in \ref{sub:algorithms}.
\else
More details and theoretical analysis on C\&CG can be found in \cite{yan_two-stage_2021,zeng2013solving}.
\fi
}

\subsection{Probabilistic Guarantees via the Scenario Approach} 
\label{sub:probabilistic_feasibility_guarantees_via_the_scenario_approach}
Throughout this paper, we construct the uncertainty set $\Delta$ using $K$ i.i.d. scenarios $\Delta = \mathcal{K} := \{\bm{\delta}^{(1)},\bm{\delta}^{(2)},\cdots,\bm{\delta}^{(K)}\}$. The resulting optimization problems are presented in \eqref{opt:two-stage-robust-single-scenario}  and \eqref{opt:two-stage-robust-two-scenario}.
Let $(\bm{x}_{\mathcal{K}}^*,\gamma_{\mathcal{K}}^*)$ denote the optimal solution to \eqref{opt:two-stage-robust-single-scenario}, e.g., returned by the C\&CG algorithm. The main results of the scenario approach theory connects the number of scenarios $K$ with the violation probability of a candidate solution.
\begin{defn}[Violation Probability]
\label{def:violation probability}
The violation probability of a candidate solution $(\bm{x}^\diamond,\gamma^\diamond)$ to \eqref{opt:two-stage-robust-single-scenario} is defined as $\mathbb{V}(\bm{x}^\diamond, \gamma^\diamond) := \mathbb{P}_{\bm{\delta}}\big( (\bm{x}^\diamond, \gamma^\diamond) \notin \mathcal{Z}(\bm{\delta}) \big)$.
\end{defn}
\begin{rem}
\label{rem:violation_probability_interpretation_general}
Mathmatically speaking, the violation probability $\mathbb{V}(\bm{x}_{\mathcal{K}}^*,\gamma_{\mathcal{K}}^*)$ quantifies the quality of the robust solution $(\bm{x}_{\mathcal{K}}^*,\gamma_{\mathcal{K}}^*)$. Specifically, $\mathbb{V}(\bm{x}_{\mathcal{K}}^*,\gamma_{\mathcal{K}}^*)$ is the probability of the following two events happening:
\begin{enumerate}
\item $\mathcal{Y}(\bm{x}_{\mathcal{K}}^*; \bm{\delta})$ is empty (infeasible second-stage problem);
\item there exists a feasible $\bm{y} \in \mathcal{Y}(\bm{x}_{\mathcal{K}}^*; \bm{\delta})$ but $\bm{d}^{\intercal} \bm{y} > \gamma_{\mathcal{K}}^*$.
\end{enumerate}
In the context of power system planning, the two events above represent two potential risks in operation, and violation probability $\mathbb{V}(\bm{x}_{\mathcal{K}}^*,\gamma_{\mathcal{K}}^*)$ has clear physical interpretations, see Section \ref{ssub:on_violation_probability} for more discussions.
\end{rem}
The main theorems of the scenario approach theory are based on the key definitions of invariant set and essential set\footnote{For the convex case in Theorem \ref{thm:apriori_guarantee_convex}, \cite{geng_computing_2019} shows that the essential set is the set of support scenarios as in \cite{campi_exact_2008}. For non-convex case, an essential set is the minimal support-subsample in \cite{campi_general_2018}.}. 

\begin{defn}[Invariant Set $\mathcal{I}$ and Essential Set $\mathcal{E}$ \cite{geng_computing_2019}]
Let $(\bm{x}_{\mathcal{I}}^*,\gamma_{\mathcal{I}}^*)$ denote the optimal solution to \eqref{opt:two-stage-robust-single-scenario} using a subset of scenarios $\mathcal{I} \subseteq \mathcal{K}$. Set $\mathcal{I}$ is an Invariant Set if $\bm{c}^{\intercal} \bm{x}_{\mathcal{I}}^* + \gamma_{\mathcal{I}}^* = \bm{c}^{\intercal} \bm{x}_{\mathcal{K}}^* + \gamma_{\mathcal{K}}^*$. An \emph{essential set} $\mathcal{E}$ is an invariant set with minimum cardinality.
\end{defn}
The original scenario approach theory \cite{campi_exact_2008,campi_wait-and-judge_2016,campi_general_2018,geng_computing_2019} only applies to single-stage optimization problems. There is a lack of known results  of multi-stage scenario approach. 
Proposition \ref{prop:two-stage-RO-as-one-stage} first shows that the two-stage robust optimization problem \eqref{opt:two-stage-robust} can be converted to an equivalent single-stage formulation \eqref{opt:two-stage-robust-single}, which enables us to extend the classical scenario approach theory towards two-stage decision making problems.
Theorems \ref{thm:apriori_guarantee_convex}, \ref{thm:aposteriori_guarantee_convex} and \ref{thm:aposteriori_guarantee_nonconvex}, which provide guarantees on the risk of the robust solution $(\bm{x}_{\mathcal{K}}^*,\gamma_{\mathcal{K}}^*)$, are essentially applying the key theorems of the classical scenario approach to the converted single-stage scenario problem\footnote{Theorem \ref{thm:apriori_guarantee_convex} is essentially Theorem 1 of \cite{campi_exact_2008}, Theorem \ref{thm:aposteriori_guarantee_convex} is from Theorem 2 of \cite{campi_wait-and-judge_2016}, and Theorem \ref{thm:aposteriori_guarantee_nonconvex} is a direct corollary of Theorem 1 in \cite{campi_general_2018}.}.
The relatively completely recourse assumption is necessary to meet the feasibility assumption in \cite{campi_exact_2008,campi_wait-and-judge_2016,campi_general_2018,geng_computing_2019}
\ifx\version\arxiv 
(see Assumption \ref{ass:feasibility} in Subsection \ref{sub:single_stage_decision_making}).
\else
(see the feasibility assumption in the appendix of \cite{yan_two-stage_2021}).
\fi

\begin{thm}[Prior Guarantees \cite{campi_exact_2008}]
\label{thm:apriori_guarantee_convex}
Suppose \eqref{opt:two-stage-robust-scenarios} is convex and has relatively complete recourse.
Given an acceptable risk level $\epsilon \in (0,1)$, a confidence parameter $\beta \in (0,1)$, and the number of first-stage decision variables $d$ of \eqref{opt:two-stage-robust-single-scenario}. Let $K$ be the smallest integer such that
\begin{equation}
\label{eqn:K_choice_convex}
\sum_{i=0}^{d-1} \binom{K}{i} \epsilon^i (1- \epsilon)^{K-i} \le \beta,
\end{equation}
then $\mathbb{P}^K\big( \mathbb{V}(\bm{x}_{\mathcal{K}}^*, \gamma^*) > \epsilon \big) \le \beta$, where $(\bm{x}^*, \gamma^*)$ is the optimal solution to \eqref{opt:two-stage-robust-single-scenario} with $K$ i.i.d. scenarios.
\end{thm}

\begin{thm}[Posterior Guarantees for Convex Problems \cite{campi_wait-and-judge_2016}]
\label{thm:aposteriori_guarantee_convex}
Suppose \eqref{opt:two-stage-robust-scenarios} is convex and has relatively complete recourse. 
Given a confidence parameter $\beta \in (0,1)$ and $(\bm{x}_{\mathcal{K}}^*, \gamma^*)$ is the optimal solution to \eqref{opt:two-stage-robust-single-scenario}. Let $|\mathcal{I}|$ be the cardinality of an invariant set, then the following probabilistic guarantee holds:
\begin{equation}
  \mathbb{P}^K \Big(\mathbb{V}(\bm{x}_{\mathcal{K}}^{*}, \gamma^*) \ge \epsilon(|\mathcal{I}|) \Big) \le \beta,
\end{equation}
where $0 < \epsilon(k) < 1$ is the (unique) solution to the polynomial equation given an integer $k=0,1,\cdots,K$,
\begin{equation}
\label{eqn:poster_convex_epsilon_function}
  \frac{\beta}{K+1} \sum_{i=k}^K \binom{i}{k} (1- \epsilon)^{i-k} - \binom{K}{k}(1- \epsilon)^{K-k} = 0
\end{equation}
\end{thm}

\begin{thm}[Probabilistic Guarantees for Non-convex Problems  \cite{campi_general_2018}]
\label{thm:aposteriori_guarantee_nonconvex}
Suppose \eqref{opt:two-stage-robust-scenarios} has relatively complete recourse, and $(\bm{x}_{\mathcal{K}}^*, \gamma^*)$ is the optimal solution to \eqref{opt:two-stage-robust-single-scenario}, let $\mathcal{I}_{\mathcal{K}}$ be an invariant set of \eqref{opt:two-stage-robust-single-scenario} using $K$ scenarios, then the following probabilistic guarantee holds:
\begin{equation}
\label{eqn:sample_complexity_nonconvex_epsilon}
\mathbb{P}^K\big( \mathbb{V}(\bm{x}_{\mathcal{K}}^*, \gamma^*) > \epsilon(|\mathcal{I}_{\mathcal{K}}|, \beta, K) \big) \le \beta,~\text{where}
\end{equation}
\begin{equation}
\epsilon(k,\beta, K) := 
  \begin{cases}
  1&~\text{if}~k=K; \\
  1 - \Big( \frac{\beta}{K \binom{K}{k}} \Big)^{\frac{1}{K-k}} &~\text{otherwise}.
  \end{cases}
\end{equation}
\end{thm}

Although Theorems \ref{thm:apriori_guarantee_convex}, \ref{thm:aposteriori_guarantee_convex} and \ref{thm:aposteriori_guarantee_nonconvex} hold for \emph{any} invariant set, the tightest guarantee is achieved with the essential set, i.e., the invariant set with minimal cardinality \cite{geng_computing_2019}.
In general, finding essential sets is a combinatorial problem, which could be computationally intractable for non-convex scenario problems \eqref{opt:two-stage-robust-scenarios}.
Proposition \ref{prop:CCG_return_invariant_set} shows that the C\&CG algorithm can effectively narrow down the range of searching.
\begin{prop}
\label{prop:CCG_return_invariant_set}
\ifx\version\arxiv 
The set $\mathcal{O}$ returned by the C\&CG algorithm (Algorithm \ref{alg:CCG}) is an invariant set.
\else
The set $\mathcal{O}$ returned by the C\&CG algorithm is an invariant set.
\fi
\end{prop}
Proposition \ref{prop:CCG_return_invariant_set} is almost self-evident. The last step of C\&CG is to solve a two-stage RO problem with all scenarios in $\mathcal{O}$. \cite{zeng2013solving} shows that C\&CG algorithm converges to an optimal solution to \eqref{opt:two-stage-robust-two-scenario}. By definition, $\mathcal{O}$ is an invariant set (not necessarily an essential one). In practice, C\&CG algorithm usually converges after only a few iterations so that the cardinality of $\mathcal{O}$ is small,
\ifx\version\arxiv 
then we can use Algorithm \ifx\version\arxiv \ref{alg:irr} \else of finding an essential set (Algorithm 3 in the appendix of \cite{yan_two-stage_2021}) \fi to identify the essential set from $\mathcal{O}$.
\else
then we can use a simple greedy algorithm (see Algorithm 3 in \cite{yan_two-stage_2021}) to identify the essential set from $\mathcal{O}$.
\fi

\section{Robust Storage Planning} 

\subsection{Compact Formulation} 
\label{sub:robust_storage_planning}
To determine the best location and size of energy storage systems, storage planning must account for short-term operational uncertainties, as the main benefits of energy storage are smoothing out the fluctuations of renewable generation and facilitating the integration of renewables. 
In the deterministic storage planning model \eqref{opt:det-SP}, wind generation $w_{n,t}=\alpha^w_{n,t}\overline{w_{n}}$ and load $d_{n,t}=\alpha^d_{n,t}\overline{d_{n}}$ are considered as deterministic trajectories. The deterministic approach \eqref{opt:det-SP} fails to take the short-term operational risk into consideration.
To account for the significant benefits of energy storage in reducing operation risk, we propose a two-stage robust storage planning model. Through constructing a scenario-based uncertainty set using wind and load data, we show that the operation risk is guaranteed to be within acceptable ranges.

The proposed two-stage robust storage planning framework is \eqref{opt-robust-compact-nonconvex}:
\begin{align}
\text{(nc-RSP)}:~\min_{ \bm{x} \in \Xi_{1}}~ C(\bm{x}) + K_d \max_{\bm{\delta} \in \Delta} \min_{ \bm{y} \in \Xi_{2}(\bm{x}; \bm{\delta})} c(\bm{y}).  \label{opt-robust-compact-nonconvex} 
\end{align}
The uncertainties are modeled by a pre-defined uncertainty set $\Delta$. There is only one difference, i.e., $\max_{\bm{\delta} \in \Delta}$, between the deterministic formulation \eqref{opt:det-SP} and its robust counterpart \eqref{opt-robust-compact-nonconvex}. By optimizing the decision for the worst-scenario, the solution to \eqref{opt-robust-compact-nonconvex} is immune against all possible realizations of uncertainties in the uncertainty set $\Delta$.

\FirstRound{The choice of the uncertainty set $\Delta$ lies at the heart of robust optimization. Throughout this paper, we construct the scenario-based uncertainty set $\mathcal{K}$ to model short-term uncertainties from renewables and loads in RSP.
\begin{equation}
\Delta = \mathcal{K} := \{\bm{\delta}^{(k)}\}_{k=1}^{|\mathcal{K}|} = \Big\{\{\alpha_{n,t}^{d,(k)}\}_{n \in \mathcal{N}, t \in \mathcal{T}},\{\alpha_{n,t}^{w,(k)}\}_{n \in \mathcal{N}, t \in \mathcal{T}} \Big\}_{k=1}^{|\mathcal{K}|}
\end{equation}
More specifically, the scenario-based uncertainty set $\mathcal{K}$ is the set of $K$ i.i.d. scenarios.
It consists of $K$ daily profiles of load factors $\{\alpha_{n,t}^{d,(k)}\}_{n \in \mathcal{N}, t \in \mathcal{T}}$ and renewable capacity factors $\{\alpha_{n,t}^{w,(k)}\}_{n \in \mathcal{N}, t \in \mathcal{T}}$.}
These scenarios could come from historical data \cite{robust_storage_in} or scenario generating algorithms \cite{data_scenario}. One direct benefit of using scenarios is to capture spatial and temporal correlations of uncertainties. More discussions on the benefits of constructing the scenario-based uncertainty set for storage planning are in Section \ref{sec:theoretical_analysis}.

In addition, these operational scenarios could come from historical data \cite{ robust_storage_in} and also can come from data-driven scenario generation method \cite{data_scenario}.  One direct benefit of using scenarios is to capture spatial and temporal correlations of uncertainties. More discussions on the benefits of constructing the scenario-based uncertainty set for storage planning are in Section \ref{sec:theoretical_analysis}.

We mainly focus on the impacts of short-term uncertainties in this paper. The major long-term uncertainties considered are load growth and increasing penetrations of renewables, which are modeled by predicted peak loads $\overline{d_{n}}$ and renewable capacities $\overline{w_{n}}$. The values of peak loads and renewable capacities could come from human experts or state-of-the-art prediction algorithms. Furthermore, the proposed framework can be easily extended towards the joint planning of energy storage, (renewable) generation, transmission lines, and other critical infrastructures.

\FirstRound{The proposed RSP model  constructs the uncertainty set from HOD to represent the future short-term uncertainties. However,  the uncertainty set only includes a small part of all possible renewable outputs and loads. It is uncertain that whether employing some scenarios from HOD is adequate to generate a robust planning result which can deal with all possible short-term uncertainties? The next section will provide a theoretical analysis for this problem.  } 

\subsection{Convex and Non-convex Formulations} 
\label{sub:convex_and_non_convex_formulations_of_storage_planning}
Due to discrete variables $z_{n,t}$ and $v_{n,t}$, the robust storage planning problem \eqref{opt-robust-compact-nonconvex} is non-convex. To reduce computational burden and obtain better theoretical results, we introduce the convexified version of \eqref{opt-robust-compact-nonconvex} in \eqref{opt-robust-compact-convex}.
We first remove the discrete planning variable $z_{n}$, thus delete constraint \eqref{opt-det:constr-first-stage} and make $P_n$ and $E_n$ planning decision variables.
\begin{align}
\underline{\rho} E_{n} \le P_{n} \le \overline{\rho} E_{n},~ n\in\mathcal{N}. \label{opt-det:constr-first-stage-relaxed}
\end{align}
\noindent Next we relax the binary variable $v_{n,t}$ (charging/discharging) to continuous variable in \eqref{eqn:relaxed_charging_discharging_no}.
\begin{equation}
\label{eqn:relaxed_charging_discharging_no}
0 \le v_{n,t} \le 1,~n \in \mathcal{N},~t \in \mathcal{T}.
\end{equation}
This relaxation is commonly adopted when studying storage system operations. Several sufficient conditions were derived to guarantee the exactness of this relaxation, and many simulation results reported that this relaxation was usually exact in practice, e.g., \cite{storage_sufficient_condition,xu_scalable_2017}.

After the modifications above, we denote the convexified feasible region as $\breve{\Xi}_1$.
\begin{equation*}
\breve{\Xi}_1 := \Big \{ (E_{n}, P_{n}, z_n)\text{s}~\text{that satisfy constraints}~\eqref{opt-det:constr-first-stage}\eqref{opt-det:constr-budget}\eqref{opt-det:constr-first-stage-relaxed} \Big \}
\end{equation*}
The feasible region of the second stage $\breve{\Xi}_2(E_{n}, P_{n}; \bm{\delta})$ is
\begin{multline*}
\breve{\Xi}_2(E_{n}, P_{n}; \bm{\delta}) := \Big \{(p_{n,t}^{\text{ch}},p_{n,t}^{\text{dis}},p_{i,n,t}^{\mathcal{G}},p_{n,t}^\text{shed})\text{s} \\
~\text{that satisfy constraints}~\eqref{opt-det:operation-constraints}  \eqref{opt-det:storage-constraints}\eqref{eqn:relaxed_charging_discharging_no} \Big \}.
\end{multline*}

The convexified robust storage planning problem is in \eqref{opt-robust-compact-convex}. Notice that both $\breve{\Xi}_1$ and $\breve{\Xi}_2(E_{n}, P_{n}; \bm{\delta})$ are convex.
\begin{align}
\text{(c-RSP)}:~\min_{ \bm{x} \in \breve{\Xi}_{1}}~ C(\bm{x}) + K_d \max_{\bm{\delta} \in \Delta} \min_{ \bm{y} \in \breve{\Xi}_{2}(\bm{x}; \bm{\delta})} c(\bm{y})  \label{opt-robust-compact-convex} 
\end{align}
In fact, the convex formulation \eqref{opt-robust-compact-convex} is also widely adopted in storage planning studies, e.g., \cite{storage_sizing_2015,storage_deg_2020}. The convex storage planning model also has two formulations: cost-minimizing with objective \eqref{opt-det:cost} and curtailment-minimizing with objective \eqref{opt-det:load-curtailment}.
In the remainder of this paper, we refer to \eqref{opt-robust-compact-nonconvex} as non-convex robust storage planning (nc-RSP), and \eqref{opt-robust-compact-convex} as convex robust storage planning (c-RSP). 
Both (c-RSP) and (nc-RSP) will be solved using the C\&CG algorithm introduced in Section \ref{sub:column_and_constraint_generation_algorithm}. 
There are four different formulations being studied in this paper. Table \ref{tab:four_formulations} provides a detailed comparison of different formulations.
\begin{table}[htbp]
  \caption{Four Different Energy Storage Formulations}
  \label{tab:four_formulations}
  \centering
\begin{footnotesize}
  \begin{tabular}{l|c|c|c|c}
  \hline

  \hline
   & \multicolumn{2}{c|}{\textbf{cost-minimizing} } & \multicolumn{2}{c}{\textbf{curtailment-minimizing} } \\
  \cline{2-5} 
   & (c-RSP) & (nc-RSP) &  (c-RSP) & (nc-RSP) \\
  \hline
  \textbf{Objective}  & \eqref{opt-det:cost} & \eqref{opt-det:cost} & \eqref{opt-det:load-curtailment} & \eqref{opt-det:load-curtailment} \\
  \hline
  \textbf{1st stage} & \eqref{opt-det:constr-first-stage}\eqref{opt-det:constr-budget}\eqref{opt-det:constr-first-stage-relaxed} & \eqref{discrete-investment:variable}\eqref{opt-det:constr-first-stage}\eqref{opt-det:constr-budget} & \eqref{opt-det:constr-first-stage}\eqref{opt-det:constr-budget}\eqref{opt-det:constr-first-stage-relaxed} & \eqref{discrete-investment:variable}\eqref{opt-det:constr-first-stage}\eqref{opt-det:constr-budget} \\
  \textbf{2nd stage} & \eqref{opt-det:operation-constraints}\eqref{opt-det:storage-constraints}\eqref{eqn:relaxed_charging_discharging_no} & \eqref{opt-det:operation-constraints}\eqref{opt-det:storage-constraints}\eqref{eqn:not_the_same_time_storage} & \eqref{opt-det:operation-constraints}\eqref{opt-det:storage-constraints}\eqref{eqn:relaxed_charging_discharging_no} & \eqref{opt-det:operation-constraints}\eqref{opt-det:storage-constraints}\eqref{eqn:not_the_same_time_storage} \\
  \hline

  \hline
  \end{tabular}
\end{footnotesize}
\end{table}

\section{Theoretical Analysis} 
\label{sec:theoretical_analysis}
\emph{A careful choice of uncertainty set $\Delta$ is critical to get meaningful storage planning results.} Throughout this paper, we adopt a data-driven approach to constructing uncertainty set $\Delta = \mathcal{K} = \{\bm{\delta}^{(k)}\}_{k=1}^{|\mathcal{K}|}$ using $K$ i.i.d. scenarios. 
This simple yet powerful choice of uncertainty set originates from the scenario approach \cite{campi_exact_2008,campi_wait-and-judge_2016,campi_general_2018}. \FirstRound{All theoretical results in this paper are based on two critical definitions: violation probability and invariant set; their mathematical definitions and interpretations in the context of RSP are discussed in Section \ref{sub:basic_definitions}.
Sections \ref{sub:prior_guarantees_RSP}-\ref{sub:posterior_guarantees_RSP} present the main theorems of this paper.}

\subsection{Basic Definitions} 
\label{sub:basic_definitions}
\subsubsection{Violation Probability $\mathbb{V}(\bm{x}^\diamond, \gamma^\diamond)$} 
\label{ssub:on_violation_probability}
\begin{defn}[Violation Probability]
\label{def:violation probability}
The violation probability of a candidate solution $(\bm{x}^\diamond,\gamma^\diamond)$ to \eqref{opt:two-stage-robust-single} is defined as $\mathbb{V}(\bm{x}^\diamond, \gamma^\diamond) := \mathbb{P}_{\bm{\delta}}\big( (\bm{x}^\diamond, \gamma^\diamond) \notin \mathcal{Z}(\bm{\delta}) \big)$.
\end{defn}

Similar with Proposition \ref{prop:two-stage-RO-as-one-stage}, we introduce an auxiliary variable $\gamma$ to denote the worst-case cost. Let $(\bm{x}^*, \gamma^*)$ be the optimal (first-stage) solution to (c-RSP) or (nc-RSP). 
In the context of RSP, $\mathbb{V}(\bm{x}^*,\gamma^*)$ depicts the probability of the following two events happening in the future:
\begin{enumerate}
\item infeasible second-stage problem (infeasible DCOPF);
\item there exists a feasible operation, but its operation cost $\gamma$ is greater than the planning solution $\gamma^*$.
\end{enumerate}
Detailed interpretations of $\mathbb{V}(\bm{x}^*,\gamma^*)$ differ in the cost-minimizing and curtailment-minimizing formulations.

\begin{rem}[Operation Cost Risk]
\label{rem:interpretation_of_in_cost_oriented_formulation}
The cost-minimizing formulation minimizes total investment and operation costs \eqref{opt-det:cost}. The solution $\bm{\gamma}^*$ is our estimation of worst-case operation costs.
As stated in Remark \ref{rem:load_shedding_give_complete_recourse}, when allowing load curtailment, the second stage is always feasible, thus $\mathbb{V}(\bm{x}^*,\gamma^*)$ only depicts event (2), i.e., the actual worst-case operation cost $\bm{d}^\intercal y(\bm{\delta}; \bm{x}^*)$ in the future is greater than $\gamma^*$:
\begin{equation}
\label{eqn:violation_prob_in_cost_RSP}
\mathbb{V}(\bm{x}^*, \gamma^*) = \mathbb{P}_{\bm{\delta}}\big( \bm{d}^\intercal y(\bm{\delta}; \bm{x}^*) > \gamma^* \big)
\end{equation}
$\mathbb{V}(\bm{x}^*, \gamma^*)$ in the cost-minimizing formulation is referred as \emph{operation cost risk}.
\end{rem}

\begin{rem}[Load Curtailment Risk]
\label{rem:interpretation_of_in_curtailment_oriented_formulation}
The curtailment-minimizing formulation minimizes total curtailment \eqref{opt-det:load-curtailment}. The solution $\gamma^*$ is our estimate on the worst-case load curtailment. $\mathbb{V}(\bm{x}^*, \gamma^*)$ is the probability that load curtailment is greater than our worst-case estimate $\gamma^*$.
\begin{equation}
\label{eqn:violation_prob_in_curtail_RSP}
\mathbb{V}(\bm{x}^*, \gamma^*) = \mathbb{P}_{\bm{\delta}}\big( \sum_{t\in \mathcal{T}}  \sum_{n\in \mathcal{N}}  p_{n,t}^\text{shed} > \gamma^* \big)
\end{equation}
$\mathbb{V}(\bm{x}^*, \gamma^*)$ in the curtailment-minimizing formulation is referred as \emph{load curtailment risk}.	
\end{rem}
We would like to point out the close relationship between $\mathbb{V}(\bm{x}^*, \gamma^*)$ and power system reliability. When $\gamma^*=0$, $\mathbb{V}(\bm{x}^*, \gamma^*)$ is the (daily) loss of load probability (LOLP). If the sole objective of storage planning is to guarantee the $\text{LOLP} \le \overline{\epsilon}$, then we can add the following constraint to (c-RSP) or (nc-RSP).
\begin{equation}
\sum_{t\in \mathcal{T}}  \sum_{n\in \mathcal{N}}  p_{n,t}^\text{shed} = 0.
\end{equation}

\subsubsection{Invariant and Essential Sets} 
\label{ssub:invariant_and_essential_sets}
\begin{defn}[Invariant Set $\mathcal{I}$ and Essential Set $\mathcal{E}$ \cite{geng_computing_2019}]
Let $(\bm{x}_{\mathcal{I}}^*,\gamma_{\mathcal{I}}^*)$ denote the optimal solution to \eqref{opt:two-stage-robust-single} using a subset of scenarios $\mathcal{I} \subseteq \mathcal{K}$. Set $\mathcal{I}$ is an \emph{invariant set} if $\bm{c}^{\intercal} \bm{x}_{\mathcal{I}}^* + \gamma_{\mathcal{I}}^* = \bm{c}^{\intercal} \bm{x}_{\mathcal{K}}^* + \gamma_{\mathcal{K}}^*$. An \emph{essential set} $\mathcal{E}$ is an invariant set with minimum cardinality.
\end{defn}
\FirstRound{An invariant set is basically the set of important scenarios that determine the optimal solution.
An essential set is essentially the minimal representation of those important scenarios.
The number of those important scenarios play a critical role in the theoretical analysis.}
\begin{prop}
\label{prop:invariant_set_convex}
The cardinality of any essential set of (c-RSP) is no greater than $2|\mathcal{S}|+1$, where $\mathcal{S}$ denotes the set of candidate locations for storage planning.
\end{prop}

Since (c-RSP) is convex, the cardinality of the essential set is bounded by the total number of decision variables of both first and second stages (Theorem 2 in \cite{calafiore_uncertain_2005}). In most cases, the number of the second-stage variables is much more than the first-stage variables. For example, the 118-bus case in Section \ref{sub:ieee_118_case_study} has $2|\mathcal{S}|+1=237$ first stage variables and $\sim 2\times 10^4$ second-stage variables. If directly applying Theorem 2 in \cite{calafiore_uncertain_2005}, we obtain a loose upper bound $|\mathcal{E}| \lesssim 2\times 10^4$, which leads to an astronomical number of scenarios $K$ per Theorem \ref{thm:a_priori_guarantee_c_RSP}. Proposition \ref{prop:invariant_set_convex} significantly tightens the bound as $|\mathcal{E}| \le 237 \ll 2\times 10^4$. 

Another attractive feature of Proposition \ref{prop:invariant_set_convex} is its independence of system size. As long as there are not too many candidate locations for storage planning, the proposed approach will not require too many scenarios when being applied on large-scale real-world systems.

\subsubsection{Relatively Complete Recourse} 
\label{ssub:relatively_complete_recourse}
\begin{rem}
\label{rem:load_shedding_give_complete_recourse}
When allowing load curtailment, both (c-RSP) and (nc-RSP) have relatively complete recourse.
\end{rem}
Note that all theorems in Sections \ref{sub:prior_guarantees_RSP}-\ref{sub:posterior_guarantees_RSP} requires the relatively complete recourse assumption. Remark \ref{rem:load_shedding_give_complete_recourse} states that our RSP formulations satisfy this assumption. 
This remark follows common sense of power system operations. When a system is in severe situations, a common control action is to shed load (e.g., rotating outages). Remark \ref{rem:load_shedding_give_complete_recourse} essentially states that system operators can shed load to a lower level to maintain minimum generation, thus the second stage problem (DCOPF) is feasible at the high cost of load curtailment. We also want to point out that Remark \ref{rem:load_shedding_give_complete_recourse} does not hold true if the following two types of constraints are included: (1) no load curtailment is allowed, i.e., $p_{n,t}^{\text{shed}}=0$; or (2) load curtailment is limited, i.e., $\sum_{n \in \mathcal{N}} p_{n,t}^{\text{shed}} \le \overline{p}^{\text{shed}}$.
All results in Sections \ref{sub:prior_guarantees_RSP}-\ref{sub:posterior_guarantees_RSP} are for the cases with load curtailment, thus meets the relatively complete recourse assumption. Similar theoretical results \emph{without} the relatively complete recourse assumption can be easily derived, e.g., using Theorem 4.1 in \cite{calafiore_random_2010}.

\subsection{A-Priori Guarantees for (c-RSP)} 
\label{sub:prior_guarantees_RSP}

\begin{thm}[A-Priori Guarantees for (c-RSP)]
\label{thm:a_priori_guarantee_c_RSP}
For the convex RSP formulation (c-RSP) with $|\mathcal{S}|$ candidate storage locations,
given an acceptable risk level $\overline{\epsilon} \in (0,1)$ and a confidence parameter $\beta \in (0,1)$, let $(\bm{x}^*, \gamma^*)$ be the optimal solution to (c-RSP) with $K$ i.i.d. scenarios, and $K$ be the smallest integer such that
\begin{equation}
\label{eqn:K_choice_c_RSP}
\sum_{i=0}^{2|\mathcal{S}|} \binom{N}{i} \overline{\epsilon}^i (1- \overline{\epsilon})^{N-i} \le \beta,
\end{equation}
then $\mathbb{P}^K\big( \mathbb{V}(\bm{x}_{\mathcal{K}}^*, \gamma^*) \le \overline{\epsilon} \big) \ge 1-\beta$.
\end{thm}

In Section \ref{sec:case_study}, after solving  hundreds of instances of (c-RSP) on different systems, we surprisingly found that the essential set of every instance of (c-RSP) was always one, i.e., $|\mathcal{E}|=1$. This feature of $|\mathcal{E}|=1$ is very appealing, \emph{as $|\mathcal{E}| = 1$ is the best non-trivial case in the sense that it requires the least amount of scenarios to achieve a given risk level $\overline{\epsilon}$}. 
One first attempt to explain the reason that $|\mathcal{I}|=1$ for (c-RSP) is Proposition \ref{prop:invariant_set_singleton}, which improves Proposition \ref{prop:invariant_set_convex} for special cases.
\begin{defn}[The Worst Scenario]
\label{defn:worst_scenario}
The \emph{worst scenario} of \eqref{opt:two-stage-robust} is defined as $\bm{\delta}^* := \arg \max_{\bm{\delta} \in \Delta} f(\bm{\delta})$, where $f(\bm{\delta}) := \min_{\bm{x} \in \mathcal{X}, (\bm{x},\gamma) \in \mathcal{Z}(\bm{\delta}) } \bm{c}^{\intercal} \bm{x} + \gamma$.
\end{defn}
\begin{prop}
\label{prop:invariant_set_singleton}
If the worst scenario $\bm{\delta}^*$ of \eqref{opt:two-stage-robust-scenarios} is unique, then $\mathcal{I} = \Omega = \{\bm{\delta}^*\}$, thus $|\mathcal{I}|=1$.
\end{prop}
Proposition \ref{prop:invariant_set_singleton} can be applied on both convex and non-convex problems. However, Proposition \ref{prop:invariant_set_singleton} needs global optimal solutions, which is challenging for non-convex problems. Table \ref{Tab:support scenario} provides an example that suboptimal solution leads to more than one element in the essential set.




\subsection{A-Posteriori Guarantees for (c-RSP) and (nc-RSP)} 
\label{sub:posterior_guarantees_RSP}
Although we cannot prove $|\mathcal{E}|=1$ with less restricted assumptions, we can still improve the a-priori guarantees (Theorem \ref{thm:a_priori_guarantee_c_RSP}) using the a-posteriori guarantees (Theorems \ref{thm:aposteriori_guarantee_convex_RSP} and \ref{thm:aposteriori_guarantee_nonconvex_RSP}). 
\FirstRound{Theorem \ref{thm:aposteriori_guarantee_convex_RSP} extends the Theorem 2 in \cite{campi_wait-and-judge_2016} towards two-stage robust optimization problems; and Theorem \eqref{thm:aposteriori_guarantee_nonconvex_RSP} is a direct corollary of Theorem 1 in \cite{campi_general_2018}.
}
\begin{thm}[A-Posteriori Guarantees for (c-RSP) \cite{campi_wait-and-judge_2016,yan_two-stage_2021}]
\label{thm:aposteriori_guarantee_convex_RSP}
Let $(\bm{x}_{\mathcal{K}}^*, \gamma^*)$ be the optimal solution to (c-RSP). Let $|\mathcal{I}_{\mathcal{K}}|$ be the cardinality of an invariant set and $\beta \in (0,1)$ be a confidence parameter chosen beforehand, then the following probabilistic guarantee holds:
\begin{equation}
  \mathbb{P}^K \Big(\mathbb{V}(\bm{x}_{\mathcal{K}}^{*}, \gamma^*) \ge \epsilon(|\mathcal{I}_{\mathcal{K}}|, \beta, K) \Big) \le \beta,
\end{equation}
where $0 < \epsilon(k, \beta, K) < 1$ is the (unique) solution to the polynomial equation given an integer $k=0,1,\cdots,K$,
\begin{equation}
\label{eqn:poster_convex_epsilon_function}
  \frac{\beta}{K+1} \sum_{i=k}^K \binom{i}{k} (1- \epsilon)^{i-k} - \binom{K}{k}(1- \epsilon)^{K-k} = 0.
\end{equation}
\end{thm}

\begin{thm}[A-Posteriori Guarantees for (nc-RSP) \cite{campi_general_2018,yan_two-stage_2021}]
\label{thm:aposteriori_guarantee_nonconvex_RSP}
Let $(\bm{x}_{\mathcal{K}}^*, \gamma^*)$ denote the optimal solution to (nc-RSP), let $\mathcal{I}_{\mathcal{K}}$ be an invariant set of (nc-RSP) using $K$ scenarios, then the following probabilistic guarantee holds:
\begin{equation}
\label{eqn:sample_complexity_nonconvex_epsilon}
\mathbb{P}^K\big( \mathbb{V}(\bm{x}_{\mathcal{K}}^*, \gamma^*) > \epsilon(|\mathcal{I}_{\mathcal{K}}|, \beta, K) \big) \le \beta,~\text{where}
\end{equation}
\begin{equation}
\label{eqn:poster_nonconvex_epsilon_function}
\epsilon(k,\beta, K) := 
  \begin{cases}
  1&~\text{if}~k=K; \\
  1 - \Big( \frac{\beta}{K \binom{K}{k}} \Big)^{\frac{1}{K-k}} &~\text{otherwise}.
  \end{cases}
\end{equation}
\end{thm}
\FirstRound{As the name suggests, a-priori guarantees (Theorem \ref{thm:a_priori_guarantee_c_RSP}) hold true \emph{before} solving the (c-RSP) problem.
In contrast, a-posteriori guarantees (Theorems \ref{thm:aposteriori_guarantee_convex_RSP} and \ref{thm:aposteriori_guarantee_nonconvex_RSP}) become valid \emph{after} obtaining the optimal solution $(\bm{x}_{\mathcal{K}}^*, \gamma^*)$ and calculating the cardinality of an invariant set $\mathcal{I}_{\mathcal{K}}$.
}
Although Theorems \ref{thm:aposteriori_guarantee_convex_RSP} and \ref{thm:aposteriori_guarantee_nonconvex_RSP} hold true for \emph{any} invariant set, the tightest guarantee is achieved with the essential set, i.e., the invariant set with minimal cardinality \cite{geng_computing_2019}.
In general, finding essential sets is a combinatorial problem, which could be computationally intractable for non-convex scenario problems \eqref{opt:two-stage-robust-scenarios}.
Proposition \ref{prop:CCG_return_invariant_set} shows that the C\&CG algorithm can effectively narrow down the range of searching.
\begin{prop}
\label{prop:CCG_return_invariant_set}
\ifx\version\arxiv 
The set $\mathcal{O}$ returned by the C\&CG algorithm (Algorithm \ref{alg:CCG}) is an invariant set.
\else
The set $\mathcal{O}$ returned by the C\&CG algorithm is an invariant set.
\fi
\end{prop}
Proposition \ref{prop:CCG_return_invariant_set} is almost self-evident. The last step of C\&CG is to solve a two-stage RO problem with all scenarios in $\mathcal{O}$. \cite{zeng2013solving} shows that C\&CG algorithm converges to an optimal solution to \eqref{opt:two-stage-robust-two-scenario}. By definition, $\mathcal{O}$ is an invariant set (not necessarily an essential one). In practice, C\&CG algorithm usually converges after only a few iterations so that the cardinality of $\mathcal{O}$ is small,
\ifx\version\arxiv 
then we can use Algorithm \ifx\version\arxiv \ref{alg:irr} \else of finding an essential set (Algorithm 3 in the appendix of \cite{yan_two-stage_2021}) \fi to identify the essential set from $\mathcal{O}$.
\else
then we can use a simple greedy algorithm (see Algorithm 3 in \cite{yan_two-stage_2021}) to identify the essential set from $\mathcal{O}$.
\fi

The main procedures to solve RSP problems and calculate theoretical guarantees are summarized in Algorithm \ref{alg:procedures_to_get_guarantees}.

\begin{algorithm}[H]
\begin{algorithmic}[1]
\STATE Choose acceptable risk level $\overline{\epsilon}$ and confidence parameter $\beta$;
\STATE Guess the cardinality of invariant set $k$;
\STATE Compute the smallest integer $K$ such that $\epsilon(k,\beta,K) \le \overline{\epsilon}$; \emph{use \eqref{eqn:poster_convex_epsilon_function} in Theorem \ref{thm:aposteriori_guarantee_convex_RSP} for (c-RSP); use \eqref{eqn:poster_nonconvex_epsilon_function} in Theorem \ref{thm:aposteriori_guarantee_nonconvex_RSP} for (nc-RSP)};
\STATE Construct the scenario-based uncertainty set $\mathcal{K} \leftarrow \{\bm{\delta}^{(k)}\}_{k=1}^{K}$ using $K$ i.i.d. scenarios;
\STATE Solve the RSP problem with the uncertainty set $\mathcal{K}$ using the C\&CG algorithm, obtain the optimal storage planning solution $(\bm{x}^*, \gamma^*)$ and an invariant set $\mathcal{O}$;
\STATE Compute an invariant set $\mathcal{I}_{\mathcal{K}} \subseteq \mathcal{O}$ as small as possible using an algorithm \ifx\version\arxiv \ref{alg:irr} \else to find essential sets (e.g., Algorithm 3 in the appendix of \cite{yan_two-stage_2021}) \fi;
\STATE Calculate $\epsilon(|\mathcal{I}_{\mathcal{K}}|, \beta, K)$; \emph{use \eqref{eqn:poster_convex_epsilon_function} in Theorem \ref{thm:aposteriori_guarantee_convex_RSP} for (c-RSP); use \eqref{eqn:poster_nonconvex_epsilon_function} in Theorem \ref{thm:aposteriori_guarantee_nonconvex_RSP} for (nc-RSP)};
\IF {$\epsilon(|\mathcal{I}_{\mathcal{K}}|, \beta, K) > \overline{\epsilon}$}
\STATE Go to step 2 with $k \leftarrow |\mathcal{I}_{\mathcal{K}}|$;
\ELSE 
\STATE Output uncertainty set $\mathcal{K}$, invariant set $\mathcal{I}_{\mathcal{K}}$, optimal solution $(\bm{x}^*, \gamma^*)$, and theoretical guarantee $\mathbb{P}^{K}(\mathbb{V}(\bm{x}^*, \gamma^*) \le \overline{\epsilon}) \ge 1 - \beta$.
\ENDIF
\end{algorithmic}
\caption{Calculating Theoretical Guarantees for RSP}
\label{alg:procedures_to_get_guarantees}
\end{algorithm}

\begin{figure*}[htbp]
\centering
\begin{tikzpicture}[node distance=1.5cm,
    every node/.style={fill=white, font=\sffamily}, align=center]
  \node (init)             [other]              {Choose risk parameters\\Initial Guess on $|\mathcal{I}_{\mathcal{K}}|$};
  \node (complexity)       [scenario,below of=init,xshift=5cm]              {Compute Sample\\Complexity $K$\\\tiny{(Equations \eqref{eqn:poster_convex_epsilon_function} or \eqref{eqn:poster_nonconvex_epsilon_function}})};
  \node (construct)        [scenario,right of=complexity,xshift=5cm]              {Construct Scenario-based\\Uncertainty Set $\mathcal{K}$\\Using $K$ i.i.d. Scenarios};
  \node (formulation)     [robust, below of=construct,yshift=-1.0cm]          {Robust Energy Storage\\Planning Formulation};
  \node (equation) 		  [below of=formulation,yshift=0.2cm] {$\min _{\bm{x} \in \mathcal{X} } \bm{c}^{\intercal} \bm{x}+\hspace{2cm}$\\$\qquad\max_{\bm{\delta} \in \mathcal{K} } \min _{\bm{y} \in \mathcal{Y}(\bm{x}; \bm{\delta})} \bm{d}^{\intercal} \bm{y}$} (formulation);
  \node (solve)      [compute, left of=formulation,xshift=-5cm]   {Compute Optimal Solution\\and Invariant Sets via\\the C\&CG Algorithm};
  \node (theory)     [scenario, left of=solve,xshift=-5cm]   {Theoretical Guarantees\\Calculate $\epsilon(|\mathcal{I}_{\mathcal{K}}|, \beta, K)$\\\tiny{(Theorems \ref{thm:aposteriori_guarantee_convex_RSP} or \ref{thm:aposteriori_guarantee_nonconvex_RSP})}};
  \node (output)      [other, below of=solve,yshift=-0.3cm]   {Output Optimal Solution $(\bm{x}^*, \gamma^*)$\\and Theoretical Guarantees\\$\mathbb{P}^{K}\Big(\mathbb{V}(\bm{x}^*, \gamma^*) \le \overline{\epsilon}\Big) \ge 1 - \beta$};
  \draw[->]     (init) -| node [xshift=-1.5cm,opacity=0,text opacity=1] {$(\overline{\epsilon},\beta)$\\$|\mathcal{I}_{\mathcal{K}}|$} (complexity);
  \draw[->]     (complexity) -- node [xshift=0cm,opacity=0,text opacity=1] {$\bm{\delta}^{(1)},\bm{\delta}^{(2)},\qquad$\\$\quad\cdots,\bm{\delta}^{(K)}$} (construct);
  \draw[->]      (construct) -- node [yshift=0.05cm,opacity=1,text opacity=1] {$\mathcal{K}:=\{\bm{\delta}^{(1)},\bm{\delta}^{(2)},\cdots,\bm{\delta}^{(K)}\}$} (formulation);
  \draw[->]     (formulation) -- (solve);
  \draw[->]      (solve) -- node [xshift=0cm,opacity=0,text opacity=1] {$(\bm{x}_{\mathcal{K}}^*, \gamma^*)$\\$\mathcal{I}_{\mathcal{K}}$} (theory);
  \draw[->]     (theory) |- node [xshift=1.8cm] {If $\epsilon(|\mathcal{I}_{\mathcal{K}}|, \beta, K) > \overline{\epsilon}$} (complexity);
  \draw[->]     (theory) |- node[xshift=1.8cm] {If $\epsilon(|\mathcal{I}_{\mathcal{K}}|, \beta, K) \le \overline{\epsilon}$} (output);
  \end{tikzpicture}
\caption{Main Procedures (Algorithm \ref{alg:procedures_to_get_guarantees}) to solve RSP problems and Calculate Theoretical Guarantees.}
\label{fig:flow_chart}
\end{figure*}
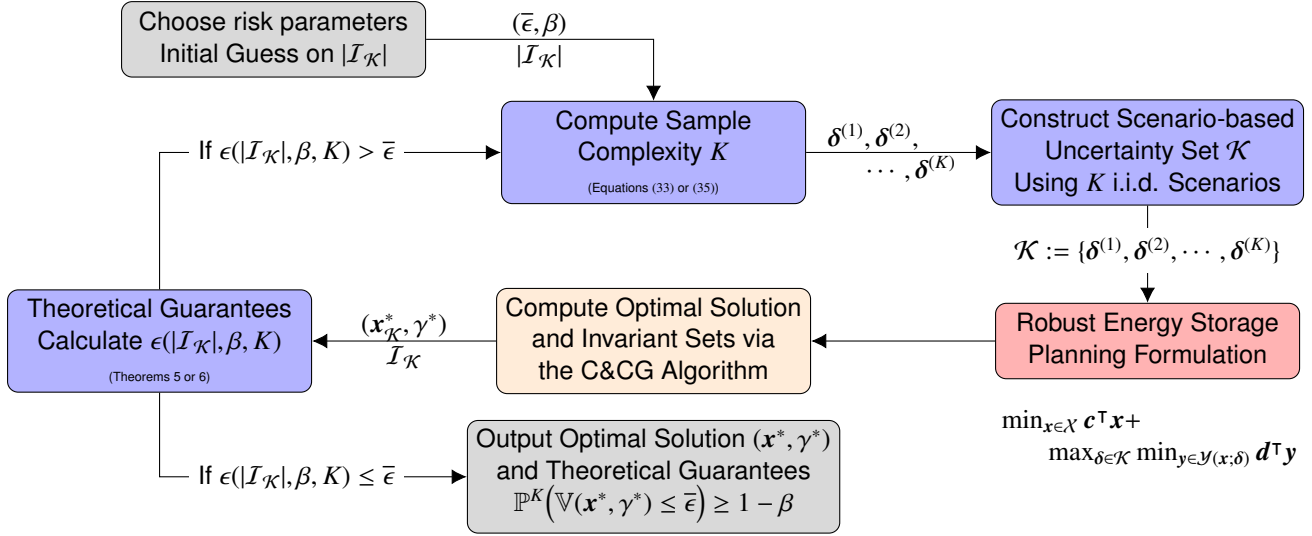

Proposition \ref{prop:CCG_return_invariant_set} states that C\&CG algorithm can identify an invariant set $\mathcal{O}$ while solving a two-stage RO problem.
After solving many instances of (c-RSP) and (nc-RSP) problems, we observed that the C\&CG algorithm usually converged within very few iterations, thus the calculated invariant sets $\mathcal{O}$ often have small cardinalities.
The next step is to apply a greedy algorithm (Algorithm \ifx\version\arxiv \ref{alg:irr} in \ref{sub:algorithms} \else 3 in the appendix of \cite{yan_two-stage_2021}\fi), checking if the removal any scenario from $\mathcal{O}$ changes optimal solution) to further pinpoint an essential set. 
Since $\mathcal{O}$ only consists of a few scenarios, \ifx\version\arxiv Algorithm \ref{alg:irr} \else this algorithm of finding an essential set \fi will stop after a few iterations. This indicates that it is computationally inexpensive to identify an essential set.


\subsection{Reduce Conservativeness and Find Better Solutions} 
\label{sub:improve_solution_quality}
The procedures described in Sections \ref{sub:prior_guarantees_RSP} and \ref{sub:posterior_guarantees_RSP} are essentially a randomized algorithm. More specifically, solutions to (c-RSP) or (nc-RSP) with different sets of scenarios of the same size $K$ could be different, but all of them possess the same theoretical guarantees. We can exploit the randomness of the algorithm to get better storage planning decisions and reduce conservativeness. For example, in Section \ref{sec:case_study}, we solved (c-RSP) or (nc-RSP) 10 times using 10 different datasets of same size, then chose the storage planning solutions with the least investment cost. This approach effectively avoids overly conservative solutions and achieves a trade-off between investment cost and risk in the robust planning model.



\section{Case Study} 
\label{sec:case_study}
Two case studies are presented in this section. The first one investigates the performance of the proposed approach in an environment with \emph{abundant data}. Uncertainties are depicted using probability distributions, from which we can sample as many scenarios as possible. The second case study is a thorough investigation of the proposed approach in a more realistic setting, i.e., a large-scale system with a limited number of available historical scenarios.
Both case studies share the same settings of energy storage systems, critical parameters are summarized in Table \ref{tab:parameters}.

\ifx\version\arxiv 
\begin{table*}[tb]
  \caption{Critical Parameters of Energy Storage Systems}
  \label{tab:parameters}
  \centering
\begin{tabular}{l|ccccccccccc}
  \hline

  \hline
  Parameter & $C_{n}^{\text{P}}$ & $C_{n}^{\text{E}}$ & $\Gamma$ & $\sigma$ & $\eta_{n}^{\text{ch}}$ & $\eta_{n}^{\text{dis}}$ & $\overline{\rho}$ & $\underline{\rho}$ \\
  \hline
  Value & $500$\$/kW & $20$\$/kWh & $10$ Years & $10\%$ & $0.9$ & $0.9$ & $0.8$ & $0.2$ \\
  \hline
  \hline
  Parameter  & $K_d$ & $c_{n}^{\text{dis}}$ &  $c_{n}^{\text{ch}}$ & $q^{\text{E}}$ & $q^{\text{P}}$ & $\overline{z}$ & $\overline{z}_n$\\
  \hline
  Value & $365$ & $18$\$/MW & $1$\$/MW & $32$MWh & $8$MW & $20$ & $4$ \\
  \hline

  \hline
  \end{tabular}
\end{table*}
\else
\begin{table*}[tb]
  \caption{Critical Parameters of Energy Storage Systems}
  \label{tab:parameters}
  \centering
\begin{footnotesize}
\begin{tabular}{l|ccccccccccc}
  \hline

  \hline
  Parameter & $C_{n}^{\text{P}}$ & $C_{n}^{\text{E}}$ & $\Gamma$ & $\sigma$ & $\eta_{n}^{\text{ch}}$ & $\eta_{n}^{\text{dis}}$ & $\overline{\rho}$ & $\underline{\rho}$ \\
  \hline
  Value & $500$\$/kW & $20$\$/kWh & $10$ Years & $10\%$ & $0.9$ & $0.9$ & $0.8$ & $0.2$ \\
  \hline
  \hline
  Parameter  & $K_d$ & $c_{n}^{\text{dis}}$ &  $c_{n}^{\text{ch}}$ & $q^{\text{E}}$ & $q^{\text{P}}$ & $\overline{z}$ & $\overline{z}_n$\\
  \hline
  Value & $365$ & $18$\$/MW & $1$\$/MW & $32$MWh & $8$MW & $20$ & $4$ \\
  \hline

  \hline
  \end{tabular}
\end{footnotesize}
\end{table*}
\fi
Since the scenario approach is a \emph{randomized} algorithm \cite{campi_exact_2008}, it is necessary to quantify the randomness of the solution returned by the scenario approach. Specifically, for a given sample complexity $K$, we repetitively solved RSP using 10 independent sets of $K$ scenarios. Each one of those 10 problems solved is referred as \emph{an experiment} in this section.

\subsection{IEEE 6-bus system}
\label{sub:ieee_6bus_case_study}
\subsubsection{System configuration} 
The first case study is based on a modified IEEE 6-bus system \cite{data_driven_planning}. We added two 100MW wind farms at buses 5 and 6, and increased the cost coefficients of coal and natural gas plants 20\% higher than the original case. Key parameters of storage investment are in Table \ref{tab:parameters}.

The uncertainties of load and wind are modeled by probability distributions. We first generated wind speed data from Weibull distributions\footnote{Key parameters: scale factor $11.0086$m/s, shape factor $1.9622$m/s, cut-in speed $V_\text{ci}=4$m/s, rated speed $V_\text{rated}=13.61$m/s, $V_\text{co}=25$m/s).}. The output of wind farms was computed based on the generated wind speed data using equation (44) in \cite{robust_storage_in}.
The load scenarios were generated by adding generated forecast errors on 365 days of historic load profiles in \cite{robust_storage_in}. We assume the load forecast errors follow a Gaussian distribution $\mathcal{N}(0,1\%)$. 
\subsubsection{Planning Results with Different Risk Parameters $\epsilon$}
\label{ssub:6bus_results_diff_epsilon}
 In  this case, we study the impact of sample complexity $K$ on the robust planning solution considering the randomness of the algorithm. We examine the costs and variances of 10 solutions of 10 experiments at the same sample complexity, i.e. the same risk level. There are 6 candidate locations in the test system, i.e. $|\mathcal{S}|=6$. When adopting the prior risk guarantees, i.e. Theorem \ref{thm:apriori_guarantee_convex},
given different risk parameters $\epsilon$, we computed the sample complexity $K$ with the same confidence parameter $\beta=10^{-3}$ in the Corollary \ref{cor:a_priori_guarantee_c_RSP}  . Results are in Table \ref{tab:priori_sample_complexity_convexity}. Given $K$ i.i.d. scenarios, we solved (c-RSP) and obtained optimal storage planning decisions. 
\begin{table}[htbp]
  \caption{The sample complexity for the different a-priori risk guarantees, using Theorem }
  \centering
\begin{footnotesize}
\begin{tabular}{l|cccccc}
  \hline

  \hline
  \textbf{A-priori risk level $\epsilon$} & \textbf{5\%} & \textbf{10\%} & \textbf{20\%} & \textbf{30\%} & \textbf{40\%} & \textbf{50\%} \\
  \hline
  \textbf{Sample complexity $K$}   & 796 & 398 & 199 & 133 & 100 & 80\\
  \hline

  \hline
  \end{tabular}
\end{footnotesize}
  \label{tab:priori_sample_complexity_convexity}
\end{table}

\begin{figure}[htbp]
  \centering
  \includegraphics[width=\linewidth]{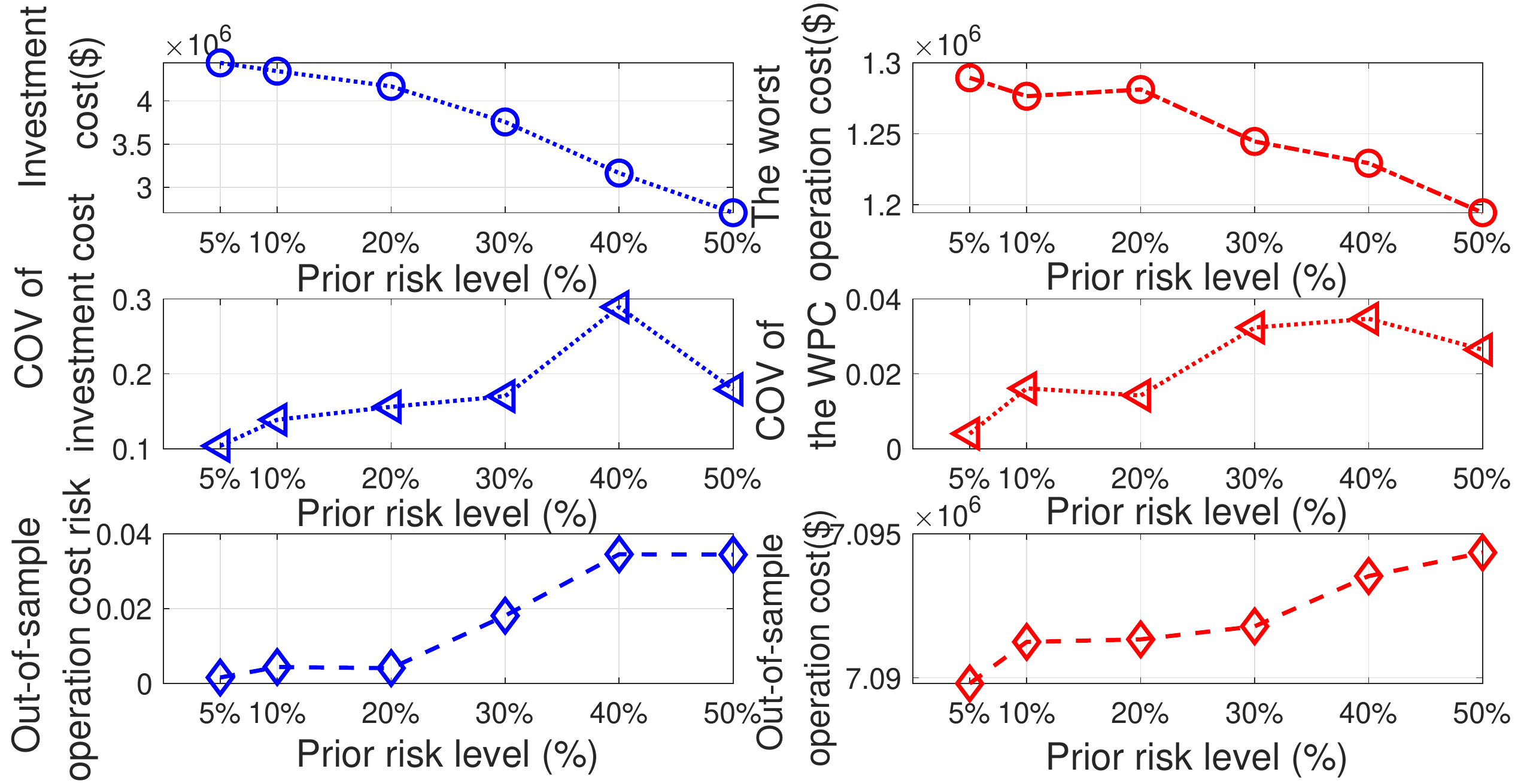}
  \caption{Storage Planning Results with Different A-priori Risk Guarantees $\epsilon$. All results were computed based on 10 independent experiments (Section \ref{ssub:6bus_results_diff_epsilon}).}
  \label{fig:multi}
\end{figure}

The first row of Fig. \ref{fig:multi} shows two components in the objective of (c-RSP): investment cost $C(\bm{x})$ and worst-case operational cost $c(\bm{y})$. Clearly planning decisions of larger risk parameters $\epsilon$ require less planning investments.
Although the top-right panel of Fig. \ref{fig:multi} also shows worst-case operational cost $c(\bm{y})$ decreases with increasing risk parameters $\epsilon$, we would like to emphasize that the worst-case operational cost $c(\bm{y})$ of (c-RSP) was evaluated on the sampled $K$ scenarios, which \emph{does not represent true operational costs}.

The true operational cost was estimated using another independent $1.6\times 10^4$ test scenarios, it is shown as the ``out-of-sample'' operation cost $c(\bm{y})$ in the bottom-right panel of Fig. \ref{fig:multi}. Clearly, the true operation cost increases with risk parameters $\epsilon$, since more load curtailment will arise.
The optimal solution $\gamma^*$ of (c-RSP) provides an estimate of worst-case operation cost. The bottom-left figure of Fig. \ref{fig:multi} shows $\mathbb{V}(\bm{x}^*, \gamma^*)$, i.e., the probability that the true operation cost is higher than $\gamma^*$ (see Remark \ref{rem:violation_probability_interpretation_general} and Section \ref{ssub:on_violation_probability}
). 

The second row of Fig. \ref{fig:multi} examines the randomness in the planning decisions. The main metric being used here is the coefficient of variance (COV), which is defined by the ratio of standard deviation of a random variable to its mean. Fig. \ref{fig:multi} illustrates the positive correlation between $\epsilon$ and COV values, more risk-averse (smaller $\epsilon$) solutions are more stable (smaller COV). "WPC" means the worst-case operation cost.
\subsubsection{Planning Results with Different Investment Budgets $C^{\text{budget}}$} 
\label{ssub:results_with_varying_investment_budgets}
Results of 10 experiments (c-RSP) with the the same risk guarantee parameters ($\epsilon=0.01$, $\beta=10^{-3}$, $K=920$) but varying budgets $C^{\text{budget}}$ are presented below.
We take a closer look at these two interesting cases in the following subsections. The case studies adopt the posterior risk guarantee for \emph{c-RSP} and correspond to Theorem \ref{thm:aposteriori_guarantee_convex} and the computation procedures in Section \ref{sub:theoretical_guarantees_cRSP}. For each result of planning experiment, we compute its out-of-sample risk level through sampling another $1.6*10^4$ scenarios from the uncertainty distribution as test data.

\begin{figure}
  \centering
  \includegraphics[width=\linewidth]{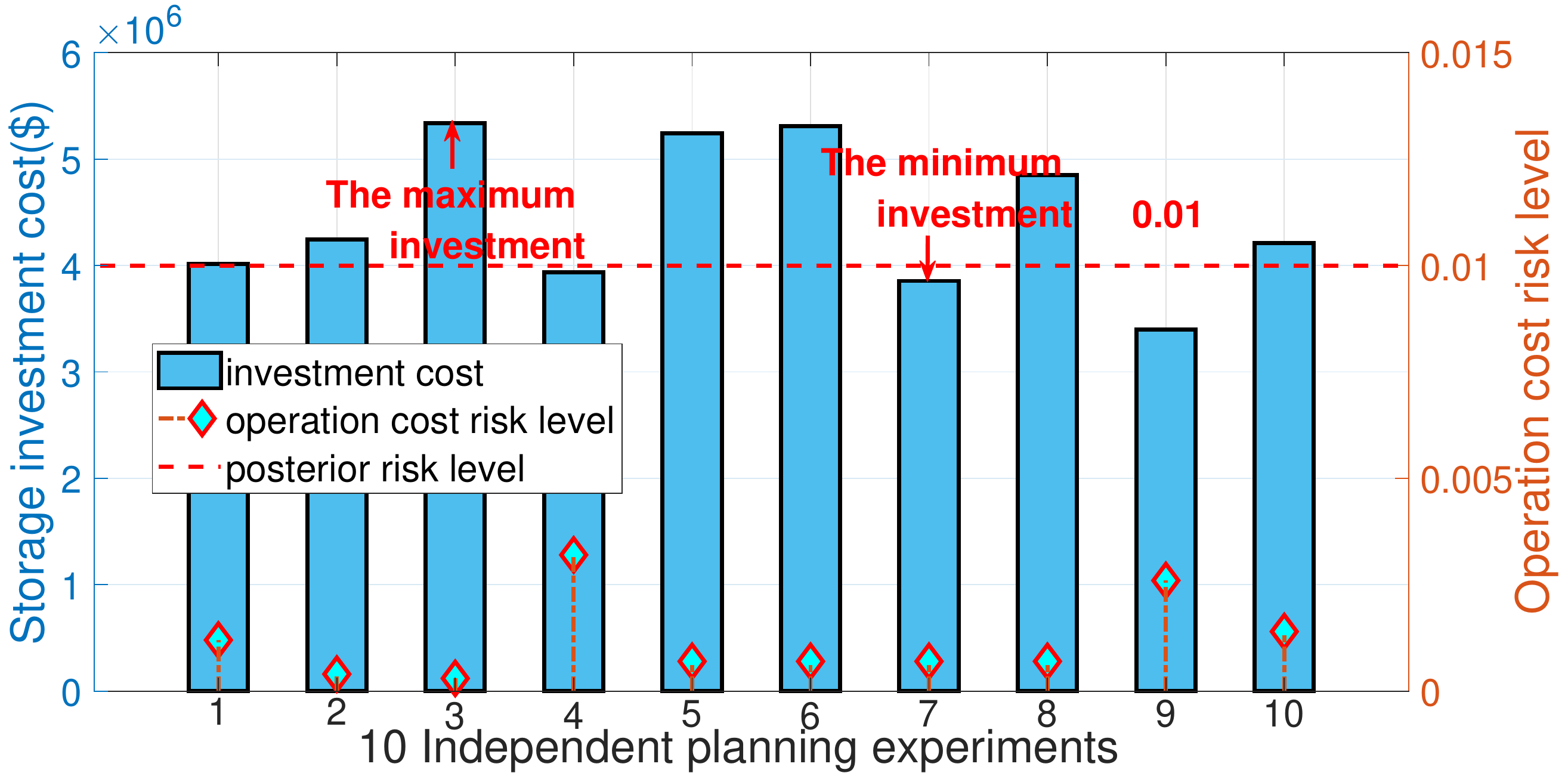}
  \caption{Storage planning results with adequate investment budget ($C(\bm{x}^*) < C^{\text{budget}}$)}
  \label{fig:6-bus-a-priori-risk}
\end{figure}

\paragraph{Adequate Investment Budget ($C(\bm{x}^*) < C^{\text{budget}}$)} 
\label{par:inactive_budget_constraints}
The purpose of robust planning is to avoid future risk by optimizing the system in the worst case. When the investment budget is adequate, i.e., inactive budget constraint $C(\bm{x}^*) < C^{\text{budget}} = \$6 \times 10^6$, the storage installment will \emph{eliminate} load curtailment in all scenarios in the uncertainty set $\mathcal{K}$.
Planning results are reported in Fig. \ref{fig:6-bus-a-priori-risk} show the storage investment costs and out-of-sample operation cost risk levels of all planning experiments. The cardinality of Invariant Set of all planning experiments are found to be 1. Since the out-of-sample violation probabilities (risk levels) are all below the risk parameter $\epsilon=0.01$ (red dotted line), we can choose the planning result with the lowest cost. This could effectively reduce the randomness and conservativeness of the scenario approach.

Despite the fact that the storage planning decisions of 10 experiments have different investment costs, they would not lead to the different actual operation cost. Actual operation cost is represented by the \emph{average} value of the operation costs of $1.6*10^4$ out-of-sample test scenarios.   The \emph{average} operation cost are almost identical, as shown in the bottom-left panel of Fig. \ref{fig:performance}. Hence, in the first row of Fig. \ref{fig:performance}, the investment costs and the total out-of-sample costs of the 10 experiments have the nearly same variation trend and amplitude. However, the bottom-right panel of Fig. \ref{fig:performance} shows that the worst-case operation cost in (c-RSP) varies due to drawing extreme scenarios. It is worth noting that the worst-case operation cost, instead of the \emph{average} operation cost, is embedded in the robust storage planning model. It is a good choice to avoid risks, but it  cannot reflect the actual operation costs. Our method can trade-off the cost and risk in the robust planning model.

\begin{figure}
\centering
\includegraphics[width=\linewidth]{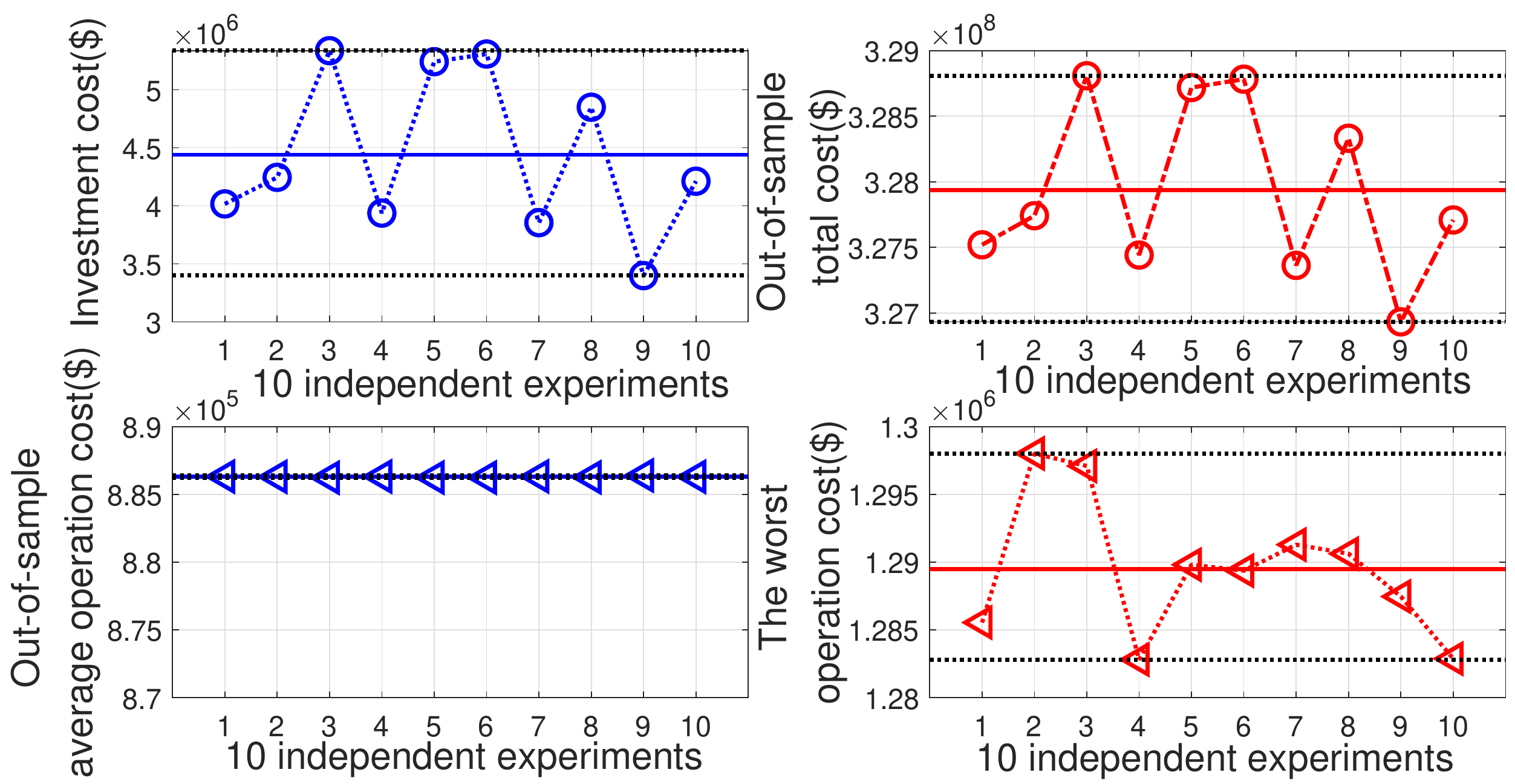}
\caption{Investment cost and operation costs with adequate investment budget ($C(\bm{x}^*) < C^{\text{budget}}$)}
\label{fig:performance}
\end{figure}


\paragraph{Inadequate Investment Budget $C(\bm{x}^*) = C^{\text{budget}}$} 
\label{par:active_budget_constraints}
When the budget is inadequate, the budget constraint is active $C(\bm{x}^*) = C^{\text{budget}}$, load curtailment are necessary in extreme scenarios. To ensure the load curtailment risk of the planning solution, we resort to the curtailment-minimizing formulation. In the first row  of Fig. \ref{fig:storage-incressing-investment-6} , the investment budget is gradually increased and the load curtailment is decreased to zero. According to the above theoretical analysis, the investment solution at $2.4\times10^6\$$ (0 load curtailment) can provide the needed load curtailment risk guarantee 0.01. 
Even though, at the investment $1.8\times10^6\$$, the out-of-sample load curtailment risk requirement seems to be satisfied. Considering the randomized property of the planning solutions, we adopt the planning solution in the investment of $2.4\times10^6\$$. Besides, all operation cost risks have been satisfied shown in the second row. 

\begin{figure}[htbp]
  \centering
  \includegraphics[width=\linewidth]{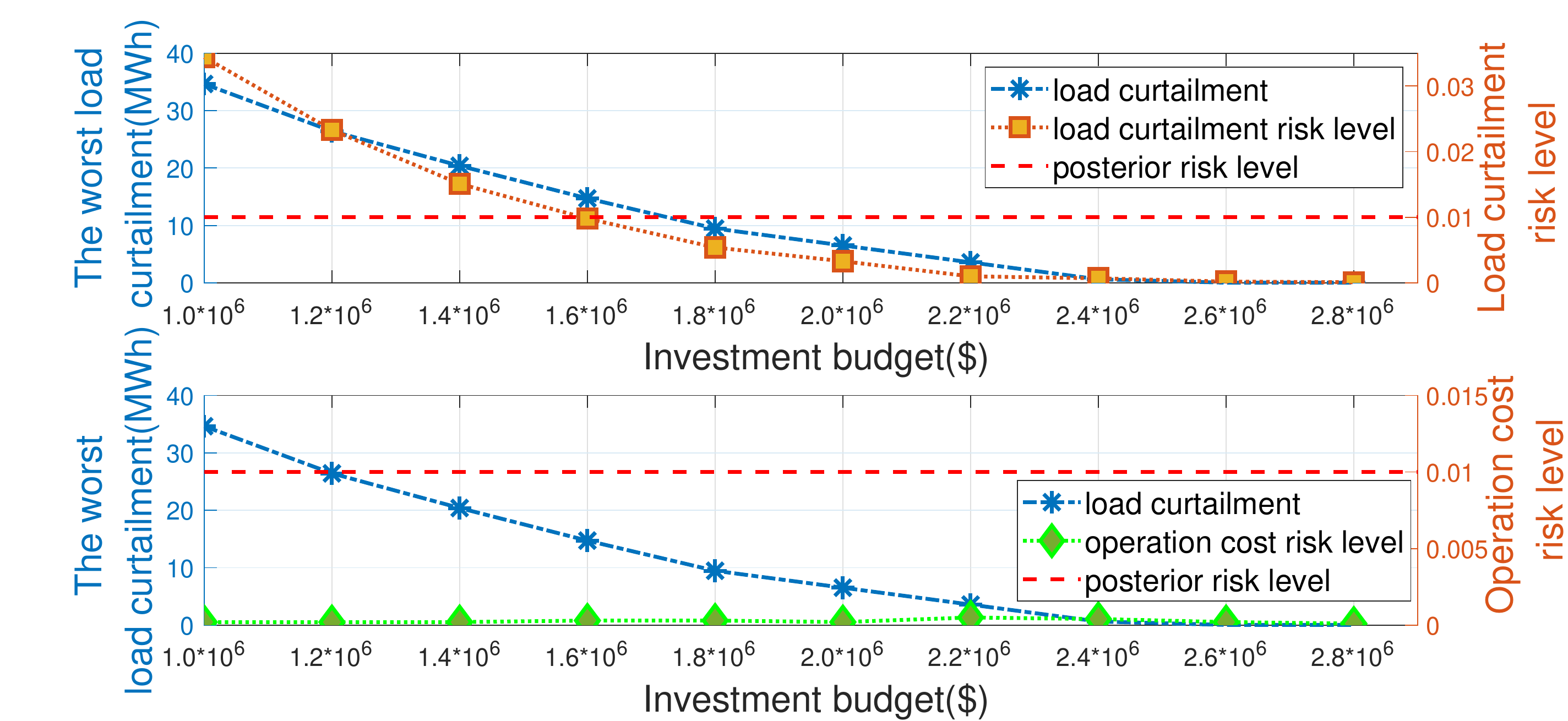}
  \caption{Storage planning with gradual increasing investments. }
  \label{fig:storage-incressing-investment-6}
\end{figure}
\subsection{IEEE 118-bus system}
\label{sub:ieee_118_case_study}
\subsubsection{System Configuration} 
\label{ssub:system_configuration}
\ifx\version\arxiv
We conducted numerical simulations on a modified 118-bus system\footnote{The original system is available at \url{http://motor.ece.iit.edu/data}. Additional changes include generation costs are set 20\% higher than the original case. The hourly ramp rates of generators are set to 45\% of the maximum generation for the largest units in the system.}
\else
We conducted numerical simulations on a modified 118-bus system\footnote{The original system is available at \url{http://motor.ece.iit.edu/data}. Complete details of the modified system are at \cite{yan_two-stage_2021}.} 
\fi
with five 200MW wind farms added to buses 16, 37, 48, 75 and 83. All 118 buses are candidate sites for storage installation. The wind and load profiles are from the Electric Reliability Council of Texas (ERCOT)\footnote{Historical load data of ERCOT is from \url{http://www.ercot.com/mktinfo/loadprofile/alp}. Wind data is from \url{http://www.ercot.com/gridinfo/resource}.}.
The dataset consists of about 7300 days (20 years) of hourly wind and load profiles.
It is worth noting that we converted the wind generation profiles to wind capacity factor $\{\alpha_{n,t}^{w,(k)}\}_{n \in \mathcal{N}, t \in \mathcal{T}}$ ($0\sim 100\%$ of full capacities), which depict the spatial and temporal patterns of wind speeds.
Load profiles were converted to load factors $\{\alpha_{n,t}^{d,(k)}\}_{n \in \mathcal{N}, t \in \mathcal{T}}$ ($0\sim 100\%$), which mainly model the temporal variations of system demands. 
The actual load and wind generation were scaled up according to the predicted peak loads and wind capacities. 
\ifx\version\arxiv
Wind capacity factors and load factors of 10 consecutive days are plotted in Figure \ref{fig:varying mode}.
\begin{figure}[htbp]
  \centering
  \includegraphics[width=\linewidth]{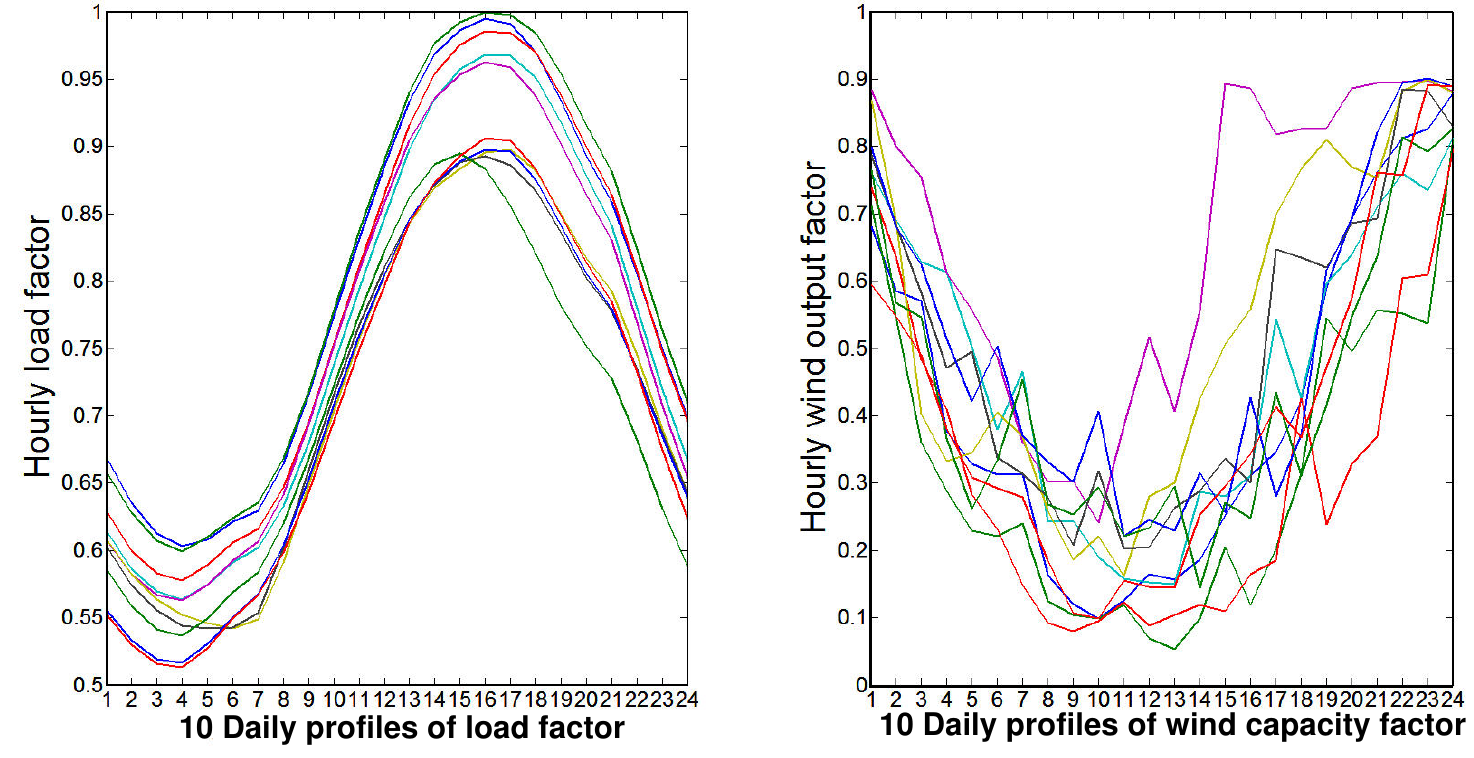}
  \centering
  \caption{10 Daily profiles of load factor and wind capacity factor in ECORT historical data}
  \label{fig:varying mode}
\end{figure}
\fi

\subsubsection{Storage Planning Results of Four Different Formulations} 
\label{ssub:storage_planning_results_of_four_different_formulations}
Similar to the 6-bus case in Section \ref{sub:ieee_6bus_case_study}, 10 experiments on the 118-bus system were conducted.
In each experiment, the 7300-day dataset was divided into two non-overlapping subsets.
The first subset contained the $K$ scenarios as input to the scenario problem. The second subset consisted of all remaining scenarios, which served as the test dataset to evaluate out-of-sample results.
\ifx\version\arxiv
If the remaining scenarios have more than 5000 scenarios, the test dataset includes 5000 scenarios; otherwise it only includes 4000 scenarios.
Storage planning results using four different formulations (see Table \ref{tab:four_formulations}) are reported in Figure \ref{fig:118-bus-results}. In all experiments, we used the same parameters $\epsilon=0.01$ and $\beta=0.001$. For non-convex formulations (nc-RSP), we used different parameters $\overline{z}_n = 10$ and $\overline{z}=1000$ from Table \ref{tab:parameters}.
\fi
We follow the procedures in Algorithm \ref{alg:procedures_to_get_guarantees} to solve (c-RSP) and (nc-RSP) problems and obtain theoretical guarantees.
\begin{figure*}[tb]
\begin{subfigure}[t]{0.49\linewidth}
  \centering
  \includegraphics[width=\linewidth]{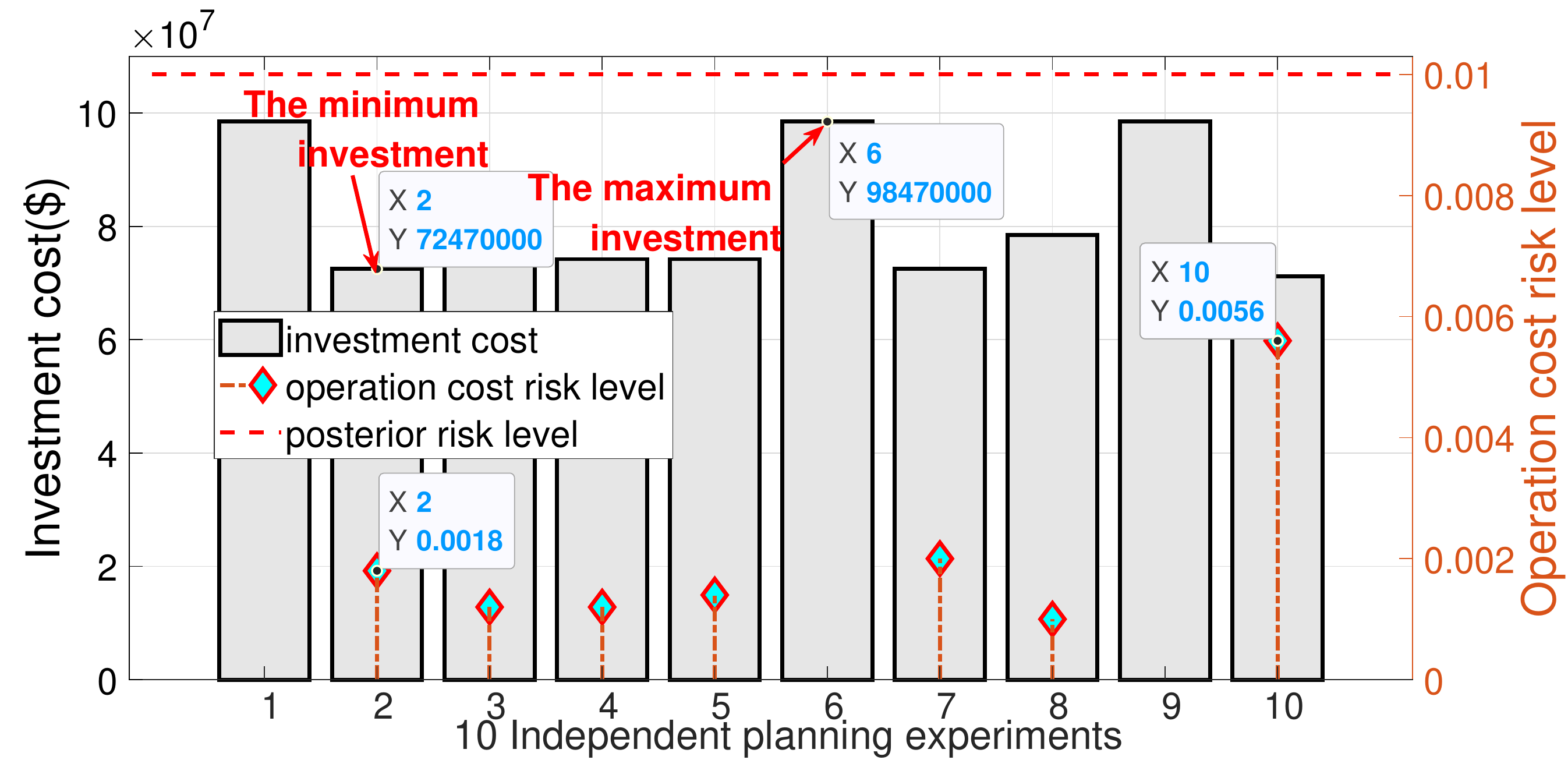}
  \caption{Cost-minimizing (c-RSP).}
  \label{fig:118-bus-continuous-investments-cost}
\end{subfigure}
\begin{subfigure}[t]{0.49\linewidth}
  \centering
  \includegraphics[width=\linewidth]{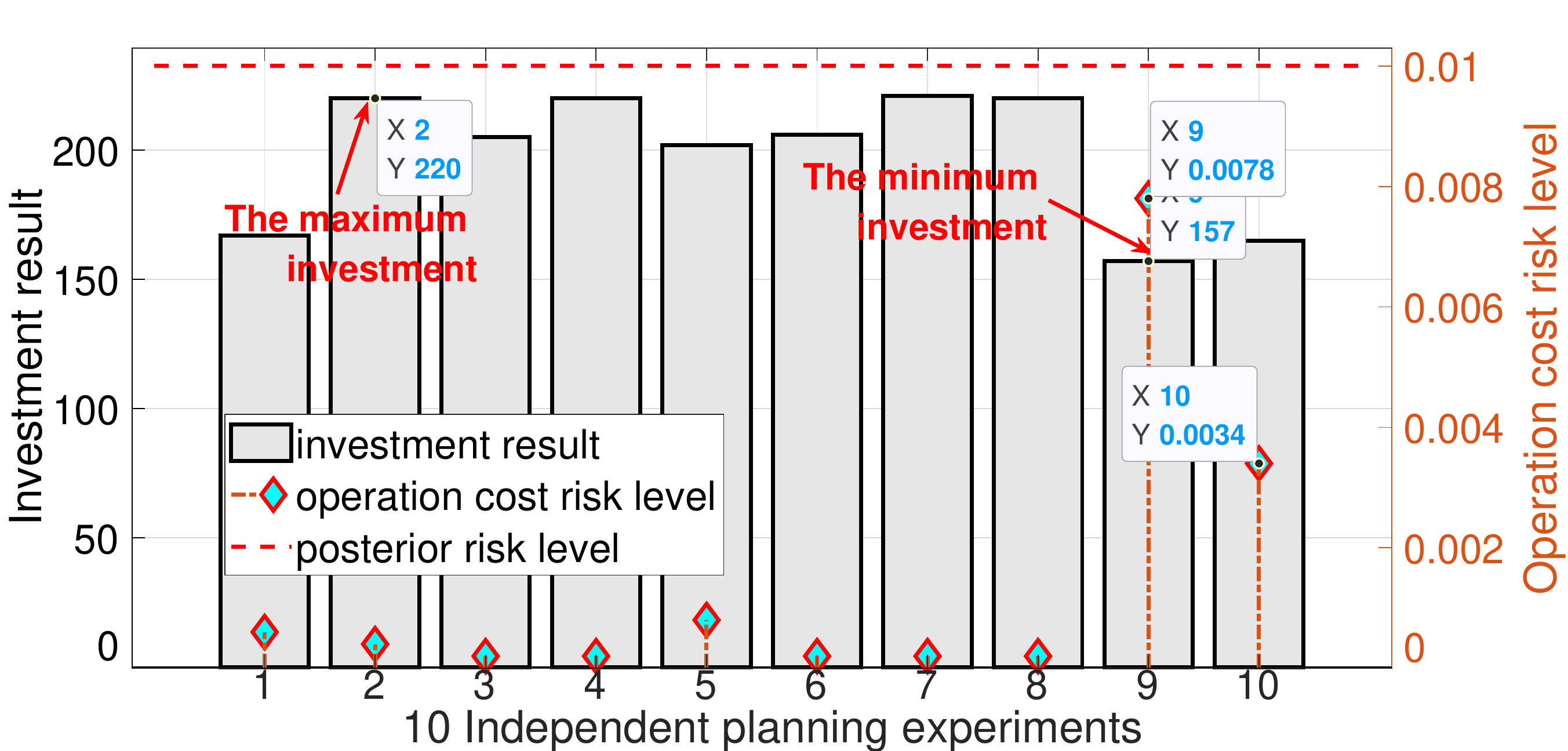}
  \caption{Cost-minimizing (nc-RSP)}
  \label{fig:118-bus-discrete-investments-cost}
\end{subfigure}
\begin{subfigure}[t]{0.49\linewidth}
  \centering
  \includegraphics[width=\linewidth]{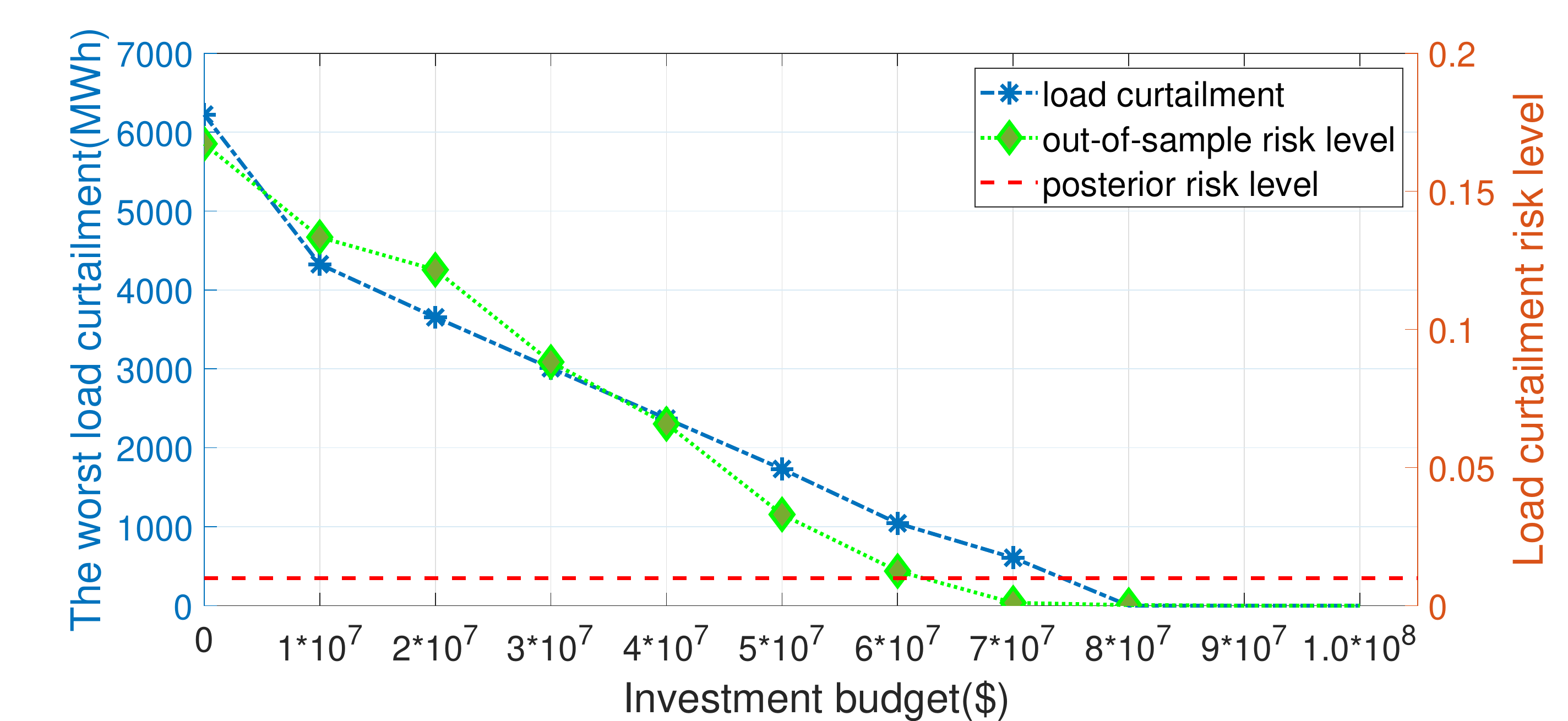}
  \caption{Curtailment-minimizing (c-RSP)}
  \label{fig:118-bus-continuous-investments-load}
\end{subfigure}
\begin{subfigure}[t]{0.49\linewidth}
  \centering
  \includegraphics[width=\linewidth]{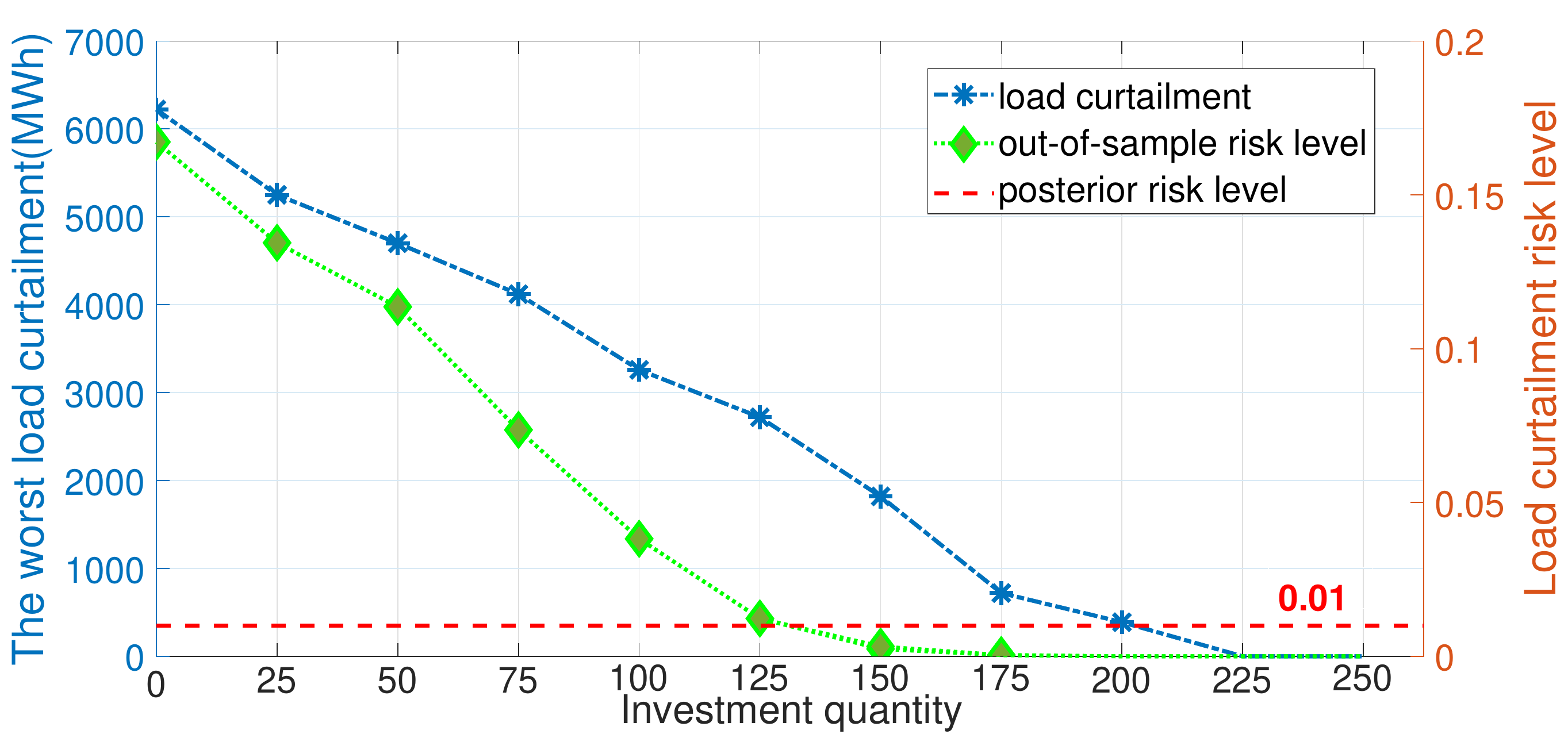}
  \caption{Curtailment-minimizing (nc-RSP)}
  \label{fig:118-bus-discrete-investments-load}
\end{subfigure}
\caption{Storage planning Results of the IEEE 118-bus system, using four different formulations \ifx\version\elsevier in Table \ref{tab:four_formulations}\fi}
\label{fig:118-bus-results}
\end{figure*}

\paragraph{Using Cost-minimizing Formulation} 
Fig. \ref{fig:118-bus-continuous-investments-cost} shows the storage planning results with adequate investment budget $C(\bm{x}) < C^{\text{budget}}=\$12\times10^7$.
No experiment reached the investment budget $C^{\text{budget}} $ and all experiments successfully maintained out-of-sample risk $\hat{\epsilon}$ within acceptable ranges ($\hat{\epsilon} < 0.01$). The largest out-of-sample risk level $\hat{\epsilon}=0.0068 \le 0.01$ happened in the $10^{th}$ planning experiment. This verifies the probabilistic guarantees in Theorem \ref{thm:aposteriori_guarantee_convex}.

Figure \ref{fig:118-bus-discrete-investments-cost} is the case that employs the non-convex storage planning model (nc-RSP), in which the capacities of energy storages are quantized. The y-axis in Figure \ref{fig:118-bus-discrete-investments-cost} is the total number of energy storage units to be installed.
We started by assuming $|\mathcal{I}|=1$ based on our observations after solving hundreds of (c-RSP) and (nc-RSP), and computed $K=2250$ using \eqref{eqn:sample_complexity_nonconvex_epsilon} to guarantee $\epsilon=0.01$ and $\beta=0.001$.
The out-of-sample violation probabilities in all 10 experiments are smaller than $\epsilon=0.01$, this verifies Theorem \ref{thm:aposteriori_guarantee_nonconvex}. Notice that the smallest total investment in storage units happened at the $9$th experiment, in which the largest out-of-sample $\hat{\epsilon} = 0.0078 < 0.01$ is within acceptable ranges. 
\ifx\version\arxiv
Notice that the $1$st and $10$th experiments have the same total storage units. This does not indicate the same storage planning solution, since they can be installed in different locations.
\fi

\paragraph{Using Curtailment-minimizing Formulation} 
Fig. \ref{fig:118-bus-continuous-investments-load} shows the convex planning result pursuing the posterior load-curtailment risk guarantee by the curtailment-minimizing formulation. To better study the impact of investment on load curtailment, in this case, the peak load level is 1.2 times of that in the above cost-minimizing planing experiment. With the increase of the investment budget, at the investment of $6 \times 10^7\$$, the out-of-sample risk level $\hat{\epsilon}$ is approaching the load curtailment risk requirement $\epsilon = 0.01$.  But considering the randomized property of the planning solution, the planning solution at the investment of $8 \times 10^7\$$ should be employed because of the 0 load curtailment is reached. The blue star is the objective function of the curtailment-minimizing formulation, i.e. the maximum load curtailment. The 0-curtailment means the risk requirement $\epsilon = 0.01$ is ensured according to the above theoretical analysis. The out-of-sample risk level (green line) based on historical data also verified the risk guarantee.

When the investment quantity is limited, the load-curtailment risk is important. Fig. \ref{fig:118-bus-continuous-investments-load} shows the non-convex planning result pursuing the posterior load-curtailment risk guarantee by the curtailment-minimizing formulation. This experiment corresponds to the above curtailment-minimizing convex planning experiment, their peak load levels and other system parameters are the same. With the increase of the investment quantity, at the investment of 150 units, the out-of-sample risk level $\hat{\epsilon}$ is below the load curtailment risk requirement $\epsilon$ 0.01. The theoretical guarantee is hold at 225 units. The blue star indicates that the objective of the curtailment-minimizing model reaches 0 at 225 units.
When we compared the non-convex and convex planning results by converting the quantity to investment cost, it is easy to find that the non-convex planning model needs more cost to reach the needed load curtailment risk guarantee. That is because the investment of quantized storage units have more physical limitations than the continuous storage units.

\FirstRound{The out-of-sample experiments of these above simulations illustrate the robust planning result has the probability guarantees for short-term uncertainties. Hence, we can employ a part of HOD to represent short-term uncertainties in the robust planning model. } 
      

\subsubsection{Convex Formulation versus Non-convex Formulation} 
\label{ssub:convex_formulation_versus_non_convex_formulation}
Although this is a much bigger system than the 6-bus case, the cardinality of essential sets of (c-RSP) is always $|\mathcal{E}| = 1$.
With the same parameters $\epsilon=0.01$ and $\beta=10^{-3}$, the 118-bus system requires the same number of scenarios $K=920$ as the 6-bus system. 
Besides the 20 experiments reported in Figures \ref{fig:118-bus-continuous-investments-cost} and \ref{fig:118-bus-continuous-investments-load}, we solved many additional (c-RSP) problems using different number of scenarios. We never found an exception of (c-RSP) with $|\mathcal{E}| > 1$.
For non-convex problems, the observation $|\mathcal{E}| \le 1$ is not always true. Most of (nc-RSP) problems have $|\mathcal{E}| \le 2$, with a few exceptions in which $|\mathcal{E}|=2$. 



\subsubsection{Obtaining Theoretical Guarantees} 
\label{ssub:finding_essential_set}
We illustrate the process of searching for essential set $\mathcal{E}$ and obtaining theoretical guarantees using one instance of (nc-RSP). 
Based on our experiences in Section \ref{ssub:convex_formulation_versus_non_convex_formulation}, we first assume $|\mathcal{E}| = 1$. Given $\epsilon=0.01$ and $\beta=0.001$, we computed $K=2220$ using equation \eqref{eqn:sample_complexity_nonconvex_epsilon}. 
We first solved (nc-RSP) with $2220$ scenarios via C\&CG algorithm. Besides the optimal storage planning solution and optimal objective value $2.4253\times 10^9 \$$, C\&CG algorithm also returned an invariant set $\mathcal{I} = \{589, 164, 1732\}$ (Proposition \ref{prop:CCG_return_invariant_set}). 
Algorithm \ifx\version\arxiv \ref{alg:irr} \else of finding an essential set (Algorithm 3 in the appendix of \cite{yan_two-stage_2021}) \fi started with $\mathcal{I} = \{589, 164, 1732\}$ and removes scenarios one by one to check if the invariant set can be future reduced.
Table \ref{Tab:support scenario} illustrates this process.
\begin{table}[htbp]
  \caption{Finding Essential Set for (nc-RSP)}
  \label{Tab:support scenario}
  \centering
\begin{footnotesize}
  \begin{tabular}{r|cccc}
  \hline

  \hline
  \textbf{Iter} & \textbf{Scenario} & \textbf{Invariant Set} & \textbf{Objective Value} & \textbf{Invariant Set} \\
  \textbf{No.} &  \textbf{Removed} & \textbf{After Removal} & \textbf{After Removal} & \textbf{Updated} \\
  \hline
  0   & -  & \{164,  589, 1732\} & \$$2.4253\times 10^9$ & \{164, 589, 1732\} \\
  1   & 164 &  \{589, 1732\} & \$$2.4253 \times 10^9$  & \{589, 1732\} \\
  2   & 589 & \{1732\} &  \$$2.2765 \times 10^9$ & \{589, 1732\} \\
  3   & 1732  & \{589\} &  \$$2.3224 \times 10^9$ & \{589, 1732\} \\
  \hline

  \hline
  \end{tabular}
\end{footnotesize}
\end{table}
The first row of Table \ref{Tab:support scenario} shows the initialization of Algorithm \ifx\version\arxiv \ref{alg:irr} \else 3 \cite{yan_two-stage_2021}) \fi. In the first iteration, scenario $164$ was removed and the optimal objective values remained unchanged, thus the updated invariant set is $\{589, 1732\}$. In the second and third iterations, removing scenario $589$ and $1732$ changed optimal solutions, therefore the invariant set $\{589, 1732\}$ cannot be further reduced.

However, notice that the invariant set $\mathcal{I} = \{589, 1732\}$ has cardinality $2$, which is greater than our initial guess $1$.
There are two possibilities of following steps. First, we can stop here if the relaxed guarantee $\epsilon(2) = $ is acceptable. 
The second choice is to recompute $K=3012$ using $\epsilon=0.01$, $\beta=0.001$ and $|\mathcal{I}|=2$, solve (nc-RSP) using another independent set of $3012$ scenarios, and check if the cardinality of the new invariant set is smaller than $2$.

In practice, we can use all available scenarios to obtain more risk-averse (i.e., smaller $\epsilon$) solutions, without reserving a major part of datasets for out-of-sample analysis. This allows us to fully exploit the value of limited data, which is one major advantage of having theoretical guarantees.

\ifx\version\arxiv
\subsection{Improve Sample Complexity} 
\label{sub:reducing_sample_complexity}
The sample complexity can be further reduced. For example, if $|\mathcal{E}| \le k$ can be \emph{proved} by the special structure of problem, according to \emph{Remark 4} in \cite{campi_general_2018}, then equation (9) in \cite{campi_general_2018} can be used to improve the sample complexity in \eqref{non-convex-number}.
If we can prove $|\mathcal{I}| \le 1$, then the sample complexity of scenarios will be reduced from $2250$ to $1410$; if $|\mathcal{I}| \le 2$ is always true, then we only need $2147$ scenarios instead of $3012$.
\fi

\section{Concluding Remarks} 
\label{sec:concluding_remarks}
This paper studies the problem of energy storage planning in future power systems through a novel data-driven scenario approach. Using the two-stage robust formulation, we explicitly account for both shorter-term fluctuations (such as during hourly operation) as well as longer-term uncertainties (such as seasonable and yearly load variations) in the storage planning problem. 

Methodologically, we connect two-stage RO optimization with the scenario approach theory to provide theoretical guarantees on the potential risk of planning solutions. The theoretical guarantees hold for both non-convex and convex planning models. We further show that the operation risk consists of two critical components: the risk of exceeding the estimated cost and the risk of load curtailment to avoid infeasible real-time operations.

In the computational aspect, we design numerical algorithms to tighten theoretical guarantees. We show that while solving RSP problems via the C\&CG algorithm, we also obtain an invariant set, which is a key component to tighten theoretical guarantees. Numerical results on the 6-bus and 118-bus systems verify the correctness of theoretical results. Due to the structure of the two-stage robust planning problem formulation, we demonstrate that the essential set is typically small and can be pinpointed at very low computational costs. Numerical results indicates that this observation is regardless of system size, which makes the proposed approach scalable and applicable for large-scale systems.

This paper is a first step towards utilizing theoretically rigorous and computationally scalable approaches to  more integrated planning decision that will address both shorter-term and longer-term uncertainties in the future grid. Future work would include: (1) developing rigorous theories which provide upper bounds on the cardinality of essential sets of two-stage robust optimization problems; (2) applying the proposed framework on the joint planning of energy storage, renewable generation, transmission, and many other critical facilities in power systems; and (3) extending the proposed planning framework towards a more detailed data-driven modeling of long-term uncertainties.



%


\ifx\version\arxiv
\appendix

\section{Additional Background and Basic Derivations} 
\label{sec:additional_background}
\subsection{Single-stage Decision Making} 
\label{sub:single_stage_decision_making}

\subsubsection{Epigraph Formulation} 
\label{ssub:epigraph_formulation}
The standard form of optimization problem is \eqref{opt:standard_opt_problem}.
\begin{subequations}
\label{opt:standard_opt_problem}
\begin{align}
\min_{\bm{x}}~& f_0(\bm{x}) \\
\text{s.t.}~& f_i(\bm{x}) \le 0,~i=1,2,\cdots,m. \\
& h_i(\bm{x}) = 0,~i=1,2,\cdots,p.
\end{align}
\end{subequations}
The \emph{epigraph formulation} of \eqref{opt:standard_opt_problem} is 
\begin{subequations}
\label{opt:epigraph_opt_problem}
\begin{align}
\min_{\bm{x},t}~& t \\
\text{s.t.}~& f_0(\bm{x}) \le t \\
& f_i(\bm{x}) \le 0,~i=1,2,\cdots,m. \\
& h_i(\bm{x}) = 0,~i=1,2,\cdots,p.
\end{align}
\end{subequations}
\begin{prop}[Chapter 4.1 of \cite{boyd_convex_2004}]
\label{prop:epigraph_formulation}
$(\bm{x}^\star, t^\star)$ is optimal for the epigraph formulation \eqref{opt:epigraph_opt_problem} if and only if $\bm{x}^\star$ is optimal for the standard form \eqref{opt:standard_opt_problem} and $t^\star = f_0(\bm{x}^\star)$.
\end{prop}
We want to emphasize the fact that \emph{Proposition \ref{prop:epigraph_formulation} holds true even when the original problem \eqref{opt:standard_opt_problem} is non-convex}. Although the proof of Proposition \ref{prop:epigraph_formulation} is elementary, we provide it below to emphasize the fact that it does not require convexity of \eqref{opt:standard_opt_problem}.
\begin{proof}[Proof of Proposition \ref{prop:epigraph_formulation}]
To prove ``$\Leftarrow$'', let  $\bm{x}^*$ denote the optimal solution to \eqref{opt:standard_opt_problem} and $t^* := f_0(\bm{x}^*)$. Clearly $(\bm{x}^*,t^*)$ is the optimal solution to \eqref{opt:epigraph_opt_problem}. Othwerwise, the optimal solution $\bm{x}^\star$ to \eqref{opt:epigraph_opt_problem} is a feasible solution to \eqref{opt:standard_opt_problem} with smaller objective, which causes contradiction. 

The proof of ``$\Rightarrow$'' is identical with above. Clearly $\bm{x}^\star$ is feasible to \eqref{opt:standard_opt_problem}, thus $f_0(\bm{x}^*) \le t^\star$. For the purpose of contradiction, we assume that $\bm{x}^* \ne \bm{x}^\star$ and $f_0(\bm{x}^*) < t^\star$. We can define $t^* := f_0(\bm{x}^*)$, then $(\bm{x}^*, t^*)$ is a feasible solution to \eqref{opt:epigraph_opt_problem} with smaller objective objective $t^* < t^\star$, which is a contradiction.
\end{proof}


\subsubsection{Chance-constrained Optimization} 
\label{ssub:chance_constrained_optimization}
A typical chance-constrained program is presented below:
\begin{subequations}
\label{opt:chance-constrained-program}
\begin{align}
\min_{\bm{x}}~& c^\intercal \bm{x} \\
\text{s.t.}~& f(\bm{x}) \le 0 \\
& \mathbb{P}_{\bm{\delta}}\big( h(\bm{x}, \bm{\delta}) \le 0 \big) \ge 1-\epsilon \label{eqn:chance}
\end{align}
\end{subequations}
where $x \in \mathbb{R}^{v}$ is the decision variable and random variable $\delta \in \Delta$ denotes the uncertainties. All deterministic constraints are represented by $f(\bm{x}) \le 0$. The chance constraint \eqref{eqn:chance} ensures the inner stochastic constraint $h(\bm{x},\bm{\delta}) \le 0$ is feasible with high probability at least $1- \epsilon$. 

\begin{defn}[Violation Probability \cite{campi_exact_2008}]
The violation probability $\mathbb{V}(\bm{x})$ of a candidate solution $\bm{x}$ is defined as
\begin{equation}
  \mathbb{V}(\bm{x}) := \mathbb{P}\big( \bm{\delta}: h(\bm{x},\bm{\delta}) > 0 \big)
\end{equation}
\end{defn}

\subsubsection{The Scenario Approach for Convex Problems} 
The scenario approach is one of the well-known solutions to chance-constrained programs \cite{geng_data-driven_2019-2}. 

To solve \eqref{opt:chance-constrained-program}, the scenario approach reformulates it to the scenario problem \eqref{opt:scenario-problem-standard} with $K$ i.i.d. scenarios $\Delta=\mathcal{K} := \{ \delta^{(1)},\delta^{(2)},\cdots,\delta^{({K})} \}$.
\begin{subequations}
\label{opt:scenario-problem-standard}
\begin{align}
\text{(SP)}_{\mathcal{K}}:~\min_{\bm{x}}~& c^\intercal \bm{x} \\
\text{s.t.}~& f(\bm{x}) \le 0 \\
&  h(\bm{x},\delta^{(1)}) \le 0, \cdots, h(\bm{x},\delta^{(K)}) \le 0
\end{align}
\end{subequations}
The scenario problem \eqref{opt:scenario-problem-standard} seeks the optimal solution $x_{\mathcal{K}}^*$ which is feasible to all $K$ scenarios. With a carefully chosen $K$, the optimal solution $\bm{x}_{\mathcal{K}}^*$ is a feasible solution to the chance-constrained program \eqref{opt:chance-constrained-program}, i.e., $\mathbb{V}(\bm{x}_{\mathcal{K}}^*) \le \epsilon$.

\begin{defn}[Support Scenario \cite{campi_exact_2008}]
\label{defn:support}
A scenario $\delta^{(i)}$ is a \emph{support scenario} for the  scenario problem $\text{SP}_{\mathcal{K}}$ if its removal changes the solution of $\text{SP}_{\mathcal{K}}$. $K$ is the number of all scenarios, $\mathcal{S}$ can denote the set of support scenarios, $|\mathcal{S}|$ is the number of support scenarios. 
\end{defn}
\begin{thm}[Theorem 2 of \cite{geng_computing_2019}]
\label{thm:number_of_support_scenarios_convex_case}
For any convex scenario problem \eqref{opt:scenario-problem-standard}, the cardinality of its invariant set is no more than the number of decision variables.
\end{thm}
\begin{defn}[Non-degenerate Scenario Problem\cite{campi_exact_2008}] Let $x_{{\mathcal{K}}}^*$ and $x_{{\mathcal{S}}}^*$ stand for the optimal solutions to the scenario problems $\text{SP}_{\mathcal{K}}$ and $\text{SP}_{\mathcal{S}}$, respectively. The scenario problem $\text{SP}_{\mathcal{K}}$ is \emph{non-degenerate} if $c^\intercal \bm{x}_{{\mathcal{K}}}^* = c^\intercal \bm{x}_{{\mathcal{S}}}^*$.
\end{defn}
\begin{assumption}[Non-Degeneracy\cite{campi_exact_2008}]
\label{ass:non_degeneracy}
For every $K$, the scenario problem \eqref{opt:scenario-problem-standard} is non-degenerate with probability 1 with respect to scenarios $\mathcal{K} := \{ \delta^{(1)},\delta^{(2)},\cdots,\delta^{({K})} \}$. 
\end{assumption}
\begin{assumption}[Feasibility]
\label{ass:feasibility}
Every scenario problem \eqref{opt:scenario-problem-standard} is feasible, and its feasibility region has a non-empty interior. The optimal solution $x_{\mathcal{K}}^*$ of \eqref{opt:scenario-problem-standard} exists. \footnote{If the problem has multiple solutions, a tie-break rule can be applied to get one unique solution.}
\end{assumption}

\begin{thm}[Prior Guarantees \cite{campi_exact_2008}]
\label{thm:a-priori scenario approach}
Under Assumptions \ref{ass:non_degeneracy} and \ref{ass:feasibility}, let $\bm{x}_{\mathcal{K}}^*$ be the optimal solution to the scenario problem $\text{SP}_{\mathcal{K}}$, it holds that
\begin{equation}
\label{a-priori scenario approach}
   \mathbb{P}^{K}\left\{\mathbb{V}\left(\bm{x}^{*}\right)>\epsilon\right\} \le \sum_{i=0}^{v-1} \binom{K}{i} \epsilon^{i}(1-\epsilon)^{K-i} 
\end{equation}
The probability $\mathbb{P}^{K}$ is taken with respect to drawing $K$ i.i.d. scenarios , and $v$ is the number of decision variables.
\end{thm}

\begin{thm}[Posterior Guarantee\cite{campi_wait-and-judge_2016}]
\label{thm:posterior_guarantee}
Given $\beta \in (0,1)$, for any $k=0,1,\cdots,K$, the polynomial equation in variable $\tau$
\begin{equation}
\label{eqn:poster_epsilon_function}
  \frac{\beta}{K+1} \sum_{i=k}^K \binom{i}{k} \tau^{i-k} - \binom{K}{k} \tau^{K-k} = 0
\end{equation}
has exactly one solution $\tau(k)$ in the interval $(0,1)$.
Let $\epsilon(k) := 1 - \tau(k)$. Under Assumptions \ref{ass:non_degeneracy} and \ref{ass:feasibility}, it holds that
\begin{equation}
  \mathbb{P}^K \Big(\mathbb{V}(\bm{x}_{\mathcal{K}}^{*}) \ge \epsilon(|\mathcal{S}|) \Big) \le \beta,
\end{equation}
where $|\mathcal{S}|$ is the number of support scenarios.
\end{thm}


\subsubsection{A General Scenario Theory for Non-Convex Problems} 
\label{sub:a_general_scenario_theory}

In general, if we extend the scenario approach theory to the general case, e.g. non-convex optimization problems. The above approaches usually become inefficient. In a prior approach, there maybe not exist a bound on $|\mathcal{S}|$ for non-convex $(\text{SP})_{\mathcal{K}}$; in a posterior approach, the \emph{non-degeneracy assumption} is crucial for the wait-and-judge approach, however, most non-convex optimization problems are not guaranteed to be a non-degenerate scenario problem. Recently, \cite{campi_general_2018} extends the scenario approach to the general case with a better bound through removing the \emph{non-degeneracy} assumption.   

For non-convex problems, it is usually computationally intractable to find global optimal solutions. There are many algorithms that are capable of finding local optimal solutions. Then, we use $\mathcal{A}_{\mathcal{K}}$ to represent a (maybe suboptimal) solution to $(\text{SP})_{\mathcal{K}}$ obtained via algorithm $\mathcal{A}$.

\begin{defn}[Support Sub-sample \cite{campi_general_2018}] Given a sample $\left(\delta^{(1)}, \cdots, \delta^{(K)}\right) \in \mathcal{K},$ a support sub-sample $\mathcal{O}$ for $\left(\delta^{(1)}, \cdots, \delta^{(K)}\right)$ is a $O$-tuple of elements extracted from $\left(\delta^{(1)}, \cdots, \delta^{(K)}\right),$ i.e. $\mathcal{O}=\left(\delta^{\left(i_{1}\right)}, \cdots, \delta^{\left(i_{O}\right)}\right)$
, which gives the same solution as the original sample, that is,
\begin{equation}
\mathcal{A}_{\mathcal{O}}\left(\delta^{\left(i_{1}\right)}, \cdots, \delta^{\left(i_{O}\right)}\right)=\mathcal{A}_{K}\left(\delta^{(1)}, \cdots, \delta^{(K)}\right)
\end{equation}
A support sub-sample $\mathcal{O}=\left(\delta^{\left(i_{1}\right)}, \cdots, \delta^{\left(i_{O}\right)}\right)$ is said to $\mathcal{O}$ leaving the solution unchanged 
be \emph{irreducible} if no element can be further removed from $\mathcal{O}$.
\end{defn}

\begin{defn}[Irreducible Set \cite{campi_general_2018}] For a support sub-sample, it is said to be \emph{irreducible} if no elements can  be removed leaving the solution unchanged.
\end{defn}

It is possible that there exist many irreducible sets for a non-convex optimization problem, for the one has the minimal cardinality, it is called the essential set in \cite{campi_general_2018}. 

\begin{thm}[Posterior Probability Guarantee on Non-convex Problem \cite{campi_general_2018}]
 Suppose that Assumption \ref{ass:feasibility} holds true, and set a value $\beta \in(0,1)$ (confidence parameter). Let $\epsilon:$
$\{0, \cdots, K\} \rightarrow[0,1]$ be a function such that
\label{thm:non-convex-scenario}
\begin{subequations}
\begin{align}
&\epsilon(K)=1 \\
&\sum_{k=0}^{K-1}\binom{K}{k}(1-\epsilon(k))^{K-k}=\beta
\label{thm:non-convex-scenario-2}
\end{align}
\end{subequations}
Then, for any $\mathcal{A}_{\mathcal{K}}$ and probability $\mathbb{P},$ it holds that
\begin{equation}
\mathbb{P}^{K}\left\{\mathbb{V}\left(x_{\mathcal{K}}^{*}\right)>\epsilon\left(s_{\mathcal{K}}^{*}\right)\right\} \le \beta
\end{equation}
where $s_{\mathcal{K}}^{*}$ is the cardinality of one support sub-sample.
\end{thm}

A simple choice of $\varepsilon(\cdot)$ obtained by splitting $\beta$ evenly among the $K$ terms in the sum \eqref{thm:non-convex-scenario-2} is:
\begin{equation}
\label{non-convex-number}
     \epsilon(k) =
\begin{cases}
1 & \text{if } k < K,\\
1-\sqrt[K-k]{\frac{\beta}{K \binom{K}{k}} } & \text{otherwise.} 
\end{cases} 
\end{equation}

To apply the probability guarantee for the non-convex problem, we first find an algorithm $\mathcal{A}$ to find the optimal solution for the non-convex problem. Then use the greedy algorithm $\mathcal{B}$ to search for a support sub-sample. The greedy algorithm in \cite{campi_general_2018} can ensure to find an irreducible set, but it needs to solve K non-convex optimization problems, it is sometimes computationally intractable. 

\begin{rem}[Remark 4 of \cite{campi_general_2018}]
Choices for $\epsilon(\cdot)$ other than \eqref{thm:non-convex-scenario-2} are possible and,
at times, advisable.
For example, if from the structure of the problem it was known that $s_K^*$ is always less than some $\overline{s}$, then it would make sense to deliberately ignore all the situations
where $s_K^* \ge \overline{s}$, thus allowing for stronger claims when $s_K^* < \overline{s}$. One possible choice is
\begin{equation}
\label{non-convex-number-improved}
     \epsilon(k) =
\begin{cases}
1 & \text{if } k < K,\\
1-\sqrt[K-k]{\frac{\beta}{\overline{s} \binom{K}{k}} } & \text{otherwise.} 
\end{cases} 
\end{equation}
\end{rem}

\section{Proofs and Algorithms} 
\label{sec:proofs_and_algorithms}
\subsection{Proofs} 
\label{sub:proofs}

There are two common formulations of robust optimization: using the \emph{robust} constraint 
\begin{subequations}
\label{opt:robust_linear_program_forall}
\begin{align}
\min_{\bm{x} \ge \bm{0}}~& \bm{c}^\intercal \bm{x} \\
\text{s.t.}~& \bm{A} \bm{x} \le \bm{d},~\forall \bm{A} \in \mathcal{U}
\end{align}
\end{subequations}
or the maximin objective
\begin{subequations}
\label{opt:robust_linear_program_maximin}
\begin{align}
\max_{\bm{A} \in \mathcal{U}}~\min_{\bm{x} \ge 0}~& \bm{c}^\intercal \bm{x} \\
\text{s.t.}~& \bm{A} \bm{x} \le \bm{d}.
\end{align}
 \end{subequations} 
Many researchers have studied the relationship between the two formulations above, e.g., \cite{beck_duality_2009,soyster_unifying_2013}. The main results is that these two formulations are equivalent if the uncertainty is constraint-wise. For robust linear programs, the constraint-wise assumption is always satisfied (see Section 2.1 of \cite{gorissen_practical_2015}), thus these two formulations are equivalent.
\begin{thm}[\cite{beck_duality_2009}, Section 8 of \cite{gorissen_practical_2015}]
\label{thm:equivalent_RLP_formulations}
The robust linear program \eqref{opt:robust_linear_program_forall} is equivalent the maximin formulation \eqref{opt:robust_linear_program_maximin}.
\end{thm}

\begin{proof}[Proof of Proposition \ref{prop:two-stage-RO-as-one-stage}]
Proposition \ref{prop:two-stage-RO-as-one-stage} can be seen as an extention of Theorem \ref{thm:equivalent_RLP_formulations} towards two-stage robust optimization problems. For a given $\bm{x}$, the second-stage problem $\max_{\bm{\delta} \in \bm{\Delta}} \min_{\bm{y} \in \mathcal{Y}(\bm{x}; \bm{\delta})} \bm{d}^{\intercal} \bm{y}$ is equivalent with 
\begin{subequations}
\label{opt:second_stage_forall}
\begin{align}
\min_{\bm{y}}~&\bm{d}^{\intercal} \bm{y} \\
\text{s.t.}~&\bm{y} \in \mathcal{Y}(\bm{x}; \bm{\delta}),~\forall \bm{\delta} \in \bm{\Delta}
\end{align}
\end{subequations}
Using the epigraph formulation (see Proposition \ref{prop:epigraph_formulation}), \eqref{opt:second_stage_forall} can be equivalently formulated as
\begin{subequations}
\label{opt:second_stage_epigraph}
\begin{align}
\min_{\bm{y},\gamma}~& \gamma \\
\text{s.t.}~&\gamma \ge \bm{d}^{\intercal} \bm{y}\\
& \bm{y} \in \mathcal{Y}(\bm{x}; \bm{\delta}),~\forall \bm{\delta} \in \bm{\Delta} 
\end{align}
\end{subequations}
Putting two stages together, we get
\begin{subequations}
\label{opt:two_stage_robust_xyr}
\begin{align}
\min_{\bm{x},\bm{y},\gamma}~&\bm{c}^{\intercal} \bm{x} + \gamma \\
\text{s.t.}~& \bm{x} \in \mathcal{X} \\
&  \gamma \ge \bm{d}^{\intercal} \bm{y}\\
& \bm{y} \in \mathcal{Y}(\bm{x}; \bm{\delta}),~\forall \bm{\delta} \in \bm{\Delta} 
\end{align}
\end{subequations}
The last step is to define $\mathcal{Z}(\bm{\delta}) := \{(\bm{x},\gamma): \exists \bm{y} \in \mathcal{Y}(\bm{x}; \bm{\delta})~\text{and}~\bm{d}^{\intercal} \bm{y} \le \gamma\}$, which is essentially projecting the original feasible set $\{(\bm{x},\bm{y},\gamma): \gamma \ge \bm{d}^{\intercal} \bm{y}~\text{and}~\bm{y} \in \mathcal{Y}(\bm{x}; \bm{\delta}),~\forall \bm{\delta} \in \bm{\Delta} \}$ onto the $\bm{x}$-space.
\end{proof}
\begin{proof}[Proof of Corollary \ref{cor:a_priori_guarantee_c_RSP}]
Corollary \ref{cor:a_priori_guarantee_c_RSP} is a direct conclusion of Theorem \ref{thm:apriori_guarantee_convex} and Proposition \ref{prop:two-stage-RO-as-one-stage}. \eqref{eqn:K_choice_c_RSP} is simply replacing d with $2|\mathcal{S}|+1$ in \eqref{eqn:K_choice_convex}.
\end{proof}

\begin{rem}
The additional variable $\gamma$ might seem unnecessary, the following formulation is equivalent without $\gamma$:
\begin{subequations}
\label{opt:two_stage_robust_xy}
\begin{align}
\min_{\bm{x},\bm{y},\gamma}~&\bm{c}^{\intercal} \bm{x} + \bm{d}^{\intercal} \bm{y} \\
\text{s.t.}~& \bm{x} \in \mathcal{X} \\
& \bm{y} \in \mathcal{Y}(\bm{x}; \bm{\delta}),~\forall \bm{\delta} \in \bm{\Delta} 
\end{align}
\end{subequations}
However, the number of support scenarios of formulation \eqref{opt:two-stage-robust-single} is much smaller than \eqref{opt:two_stage_robust_xy}.
\end{rem}


\begin{proof}[Proof of Proposition \ref{prop:invariant_set_convex}]
For the convex formulation (c-RSP), there are only two first-stage decision variables $E_n$ and $P_n$. Using Proposition \ref{prop:two-stage-RO-as-one-stage}, (c-RSP) can be equivalently formulated as a single-stage problem \eqref{opt:two-stage-robust-single} with an additional variable $\gamma$. This new formulation does not have any second-stage decision variables $\bm{y} = (p_{n,t}^{\text{ch}},p_{n,t}^{\text{dis}},p_{i,n,t}^{\mathcal{G}},p_{n,t}^\text{shed}, v_{n,t})$. The last step is to apply Theorem \ref{thm:number_of_support_scenarios_convex_case} on \eqref{opt:two-stage-robust-single}, which consists of $2|\mathcal{S}|+1$ decision variables ($E_n$, $P_n$, and $\gamma$).
\end{proof}

\begin{proof}[Proof of Proposition \ref{prop:invariant_set_singleton}]
When $\bm{\delta}^*$ is the unique solution to $\arg \max_{\bm{\delta} \in \mathcal{K}} f(\bm{\delta})$, simply removing $\bm{\delta}^*$ from $\mathcal{K}$ will leads to smaller optimal objective
\begin{equation}
\max_{\bm{\delta} \in \mathcal{K}} f(\bm{\delta}) > \max_{\bm{\delta} \in \mathcal{K} - \bm{\delta}^* } f(\bm{\delta})
\end{equation}
Thus $\bm{\delta}^*$ is clearly a support scenario. Also note that removing any other scenarios in $\mathcal{K} - \bm{\delta}^*$ will not change the optimal solution to $\max_{\bm{\delta} \in \mathcal{K}} f(\bm{\delta})$. This concludes the proof.
\end{proof}

\subsection{Algorithms} 
\label{sub:algorithms}

\begin{algorithm}[H]
\begin{algorithmic}[1]
\STATE set $\text{LB}=-\infty$,  $\text{UB}=+\infty$, $\mathcal{K} := \{\bm{\delta}^{(1)},\bm{\delta}^{(2)},\cdots,\bm{\delta}^{(K)}\}$, $\mathcal{O}=\emptyset$ and $k=0$.
\WHILE{$\text{UB}-\text{LB} \le \varepsilon$}
\STATE solve the master problem \eqref{master-problem},
\begin{subequations}
\label{master-problem}
\begin{align}
\min _{\gamma,\bm{x} \in \mathcal{X} }~ &\bm{c}^{\intercal} \bm{x}+ \gamma\\
\text{s.t.}~&\bm{d}^{\intercal} \bm{y}^{(l)} \leq \gamma,~l=1,2,\cdots,k \\
& \bm{y}^{(l)}\in \mathcal{Y}(\bm{x}, \bm{\delta}^{*{(l)}}),~l=1,2,\cdots,k
\end{align}
\end{subequations}
obtain an optimal solution $(\bm{x}_{k+1}^*, \gamma_{k+1}^*, \bm{y}^{1*},\cdots,\bm{y}^{k*})$.
\STATE update $\text{LB} \leftarrow \bm{c}^{\intercal} \bm{x}^*_{k+1}+\gamma^*_{k+1}$.
\STATE solve the sub-problem \eqref{sub-problem} with $\bm{x}_{k+1}^*$,
\begin{equation}
\label{sub-problem}
  \max_{\bm{\delta} \in \mathcal{K} } \min _{\bm{y} \in \mathcal{Y}(\bm{x}_{k+1}^*, \bm{\delta})} \bm{d}^{\intercal} \bm{y} 
\end{equation}

\STATE if \eqref{sub-problem} is feasible, obtain the optimal solution $\bm{y}_{k+1}^*$ and $\bm{\delta}^{*(k+1)}$. Create variable $\bm{y}^{(k+1)}$, add the following constraints to \eqref{master-problem}.
\begin{subequations}
\label{feasible-problem}
\begin{align}
&\bm{d}^{\intercal} \bm{y}^{(k+1)} \leq \gamma \\
& \bm{y}^{(k+1)}\in \mathcal{Y}(\bm{x}, \bm{\delta}^{*{(k+1)}})
\end{align}
\end{subequations}

\STATE if \eqref{sub-problem} is infeasible, $\bm{c}^{\intercal} \bm{x}^*_{k+1}+\bm{d}^{\intercal} \bm{y_{k+1}}^*=+\infty$. Create variable $\bm{y}^{(k+1)}$, add the following constraints to \eqref{master-problem}.
\begin{equation}
\label{infeasible-problem}
 \bm{y}^{(k+1)}\in \mathcal{Y}(\bm{x}, \bm{\delta}^{*{(k+1)}})
\end{equation}
where $\bm{\delta}^{*{(k+1)}}$ is the identified scenario for the infeasible \eqref{sub-problem}.
\STATE update $\text{UB} \leftarrow \min(\text{UB},\bm{c}^{\intercal} \bm{x}^*_{k+1}+\bm{d}^{\intercal} \bm{y_{k+1}}^*)$. 
\STATE update $k$ and $\mathcal{O}$ as $\mathcal{O}\leftarrow\{\mathcal{O},\delta^{*(k+1)}\}$, $k \leftarrow k+1$;
\ENDWHILE
\RETURN $\bm{x^*}$, $\bm{y}^*$ and $\mathcal{O}$.
\end{algorithmic}
\caption{column-and-constraint generation (C\&CG) \cite{zeng2013solving}}
\label{alg:CCG}
\end{algorithm}

\begin{algorithm}
 \caption{Find an Essential Set for the Robust Storage Planning}
 \label{alg:irr}
 \begin{algorithmic}[1]
 \renewcommand{\algorithmicrequire}{\textbf{Input:}}
 \renewcommand{\algorithmicensure}{\textbf{Output:}}
 \REQUIRE  $\mathbf{\Delta}=\mathcal{K}=\{\delta^{(1)},...,\delta^{(K)}\}$;
 \\ \textit{Find an Invariant Set} :
  \STATE {Solve (c-RSP) or (nc-RSP) with scenario-based uncertainty set $\mathcal{K}$ by C\&CG algorithm, and obtain the set of scenarios  in the master problem $\mathcal{O}$}.
  \STATE {Delete the repetitive scenarios $\delta^{*}$ in $\mathcal{O}$ , and obtain an Invariant Set $\mathcal{O}=\{\delta^{(i_1)},...,\delta^{(i_{|\mathcal{O}|})}\}$}
   \\ \textit{Find an Essential Set}
  \STATE {$\mathcal{I}_{\mathcal{K}}\leftarrow\mathcal{O}$}; solve (c-RSP) or (nc-RSP) with the uncertainty set $\mathcal{O}$ and obtain the solution $x_{\mathcal{O}^{*}}$;
  \FOR {o=1 to $|\mathcal{O}|$} 
     \STATE      $\overline{\mathcal{I}}_{\mathcal{K}}= \mathcal{I}_{\mathcal{K}}\text{\textbackslash} ~ \delta^{(i_o)}$; solve (c-RSP) or (nc-RSP) with the uncertainty set $\overline{\mathcal{I}}_{\mathcal{K}}$ and obtain the solution $x_{\overline{\mathcal{I}}_{\mathcal{K}}}^{*}$ .
  \IF {$x_{\overline{\mathcal{I}}_{\mathcal{K}}}^{*}=x_{\mathcal{O}}^{*}$} 
  \STATE {$\mathcal{I}_{\mathcal{K}}\leftarrow{\overline{\mathcal{I}}_{\mathcal{K}}}$}.
 \ENDIF
 \ENDFOR
 \RETURN {Essential Set $\mathcal{I}_{\mathcal{K}}$ and its cardinality $|\mathcal{I}_{\mathcal{K}}|$. }
 \end{algorithmic} 
 \end{algorithm}

\fi

\ifx\version\arxiv
\bibliographystyle{elsarticle-num}
\else
\bibliographystyle{elsarticle-num}
\fi
\bibliography{references}

\end{document}